\documentclass[journal,onecolumn,draftcls]{IEEEtran}
\newif\iffull
 \fulltrue

\usepackage[english]{babel}
\usepackage{amssymb}
\usepackage{amsmath}
\usepackage[dvipsnames]{xcolor}
\usepackage{cite}
\usepackage{graphicx}
\usepackage{balance}
\usepackage{amsthm}
\usepackage{comment}
\usepackage[hidelinks]{hyperref}
\usepackage{mathrsfs}
\usepackage{tikz}
\usetikzlibrary{automata,positioning}

\theoremstyle{definition}
\newtheorem{theorem}{Theorem}
\newtheorem{lemma}{Lemma}
\newtheorem{remark}{Remark}
\newtheorem{construction}{Construction}
\newtheorem{corollary}{Corollary}
\newtheorem{definition}{Definition}

\newtheorem{example}{Example}
\newtheorem{observation}{Observation}
\newtheorem{problem}{Problem}
\newtheorem{claim}{Claim}

\addto\extrasenglish{

}

\newcommand{\N}{\mathbb{N}}

\newcommand{\cC}{{\cal C}}

\newcommand{\cE}{{\cal E}}

\newcommand{\cN}{{\cal N}}
\newcommand{\cO}{{\cal O}}
\newcommand{\cP}{{\cal P}}

\newcommand{\cR}{{\cal R}}

\newcommand{\cV}{{\cal V}}

\newcommand{\bfw}{{\boldsymbol w}}

\newcommand{\mA}{{\mathsf A}}
\newcommand{\mC}{{\mathsf C}}
\newcommand{\mG}{{\mathsf G}}
\newcommand{\mT}{{\mathsf T}}
\newcommand{\mN}{{\mathsf N}}

\newcommand{\DNA}{\cE}

\def\CC#1#2{\ensuremath{\left(\kern-.3em\left(\genfrac{}{}{0pt}{}{#1}{#2}\right)\kern-.3em\right)}}

\newcommand{\set}[2]{\left\{#1\;\left|\; #2\right.\right\}}
\newcommand{\abs}[1]{\left|#1\right|}

\newcommand{\te}[1]{{\footnotesize  [{\textcolor{blue}{#1}} \textcolor{blue!60!black}{--Tuvi}]\normalsize}}
\newcommand{\db}[1]{{\footnotesize  [{\textcolor{green!50!black}{#1}} \textcolor{purple!60!black}{--Daniella}]\normalsize}}

\newcommand{\full}[1]{\iffull #1 \fi}

\title{Representing Information on DNA using Patterns Induced by Enzymatic Labeling 
}
\author{\IEEEauthorblockN{Daniella Bar-Lev\IEEEauthorrefmark{1}\IEEEauthorrefmark{2},  Eitan Yaakobi\IEEEauthorrefmark{2}, Zohar Yakhini\IEEEauthorrefmark{2}\IEEEauthorrefmark{3}, and Tuvi Etzion\IEEEauthorrefmark{2}}\\
\thanks{The work of D. Bar-Lev, E. Yaakobi, and Z. Yakhini was funded by the European Union (DiDAX, 101115134). Views and opinions expressed are however those of the author(s) only and do not necessarily reflect those of the European Union or the European Research Council Executive Agency. Neither the European Union nor the granting authority can be held responsible for them. The work of D. Bar-Lev was also supported in part by Schmidt Science and by NSF Grant CCF-2212437. 
The research of T. Etzion and E. Yaakobi was also partially funded by the Israeli Science Foundation 2462/24.}
\thanks{Part of this work was presented and published at the IEEE International Symposium on Information Theory (ISIT), Athens, Greece, 2024~\cite{barlev24isit}.}
\IEEEauthorblockA{\IEEEauthorrefmark{1}\textit{Department of Mathematics, Universität Zürich, Switzerland}}\\
\IEEEauthorblockA{\IEEEauthorrefmark{2}\textit{Faculty of Computer Science, Technion -- Israel Institute of Technology, Haifa, Israel}} \\
\IEEEauthorblockA{\IEEEauthorrefmark{3}\textit{School of Computer Science, Reichman University, Herzliya, Israel}\\
Email: \{daniellalev, yaakobi, zohary, etzion\}@cs.technion.ac.il}}

\IEEEoverridecommandlockouts
\IEEEaftertitletext{\vspace{-5pt}}

\begin{document}

\maketitle
\vspace{-2ex}
\begin{abstract}
Enzymatic DNA labeling is a powerful tool 
with applications in biochemistry, molecular biology, biotechnology, medical science, and genomic research. This paper contributes to the evolving field of DNA-based data storage by presenting a formal framework for modeling DNA labeling in strings, specifically tailored for data storage purposes. Our approach involves a known DNA molecule as a template for labeling, employing patterns induced by a set of designed labels to represent information. 
One hypothetical implementation can use CRISPR-Cas9 and gRNA reagents for labeling.
Various aspects of the general labeling channel, including fixed-length labels, are explored, and upper bounds on the maximum size of the corresponding codes are given. The study includes the development of an efficient encoder-decoder pair that is proven to be optimal in terms of maximum code size under specific conditions.

The codewords for this framework are codesets, i.e., each codeword is a collection of a few labels. Of special interest are codes whose codesets contain fixed-length labels. Their construction and analysis are tied to de Bruijn sequences and constrained coding.
\end{abstract}

\vspace{-1.9ex}
\section{Introduction}
\label{sec:introduction}
Enzymatic DNA labeling is a powerful tool 
with applications in biochemistry, molecular biology, biotechnology, medical science, and genomic research~\cite{
moter2000fluorescence, chen2018efficient, gruszka2021single,
ma2015multicolor,
deen2017methyltransferase,
young2020technical, assa2024quantifying}. The technique involves the deliberate labeling of specific parts of the DNA molecule with specific markers. In recent literature, it was also demonstrated that DNA labeling (and more generally, DNA editing) can be used in the emerging technology of data storage on DNA, e.g., \cite{tabatabaei2020dna, Sadremomtaz2023storageCRISPR}. One important approach to labeling is the use of CRISPR-Cas9 based systems \cite{jinek2012programmable, chen2018efficient, cong2013multiplex, amit2021crispector}. 

CRISPR-Cas9 allows for precise targeting and modification of DNA sequences. The system consists of two key components: a guide RNA (gRNA) that directs the Cas9 enzyme to a specific DNA sequence, and the Cas9 enzyme itself, which introduces a cut or modification at the targeted site. High precision makes CRISPR-Cas9 an ideal candidate for DNA labeling, as it enables the attachment of markers or tags to specific locations on the DNA strand. However, the physical constraints of the enzyme and the labeling process impose limitations, such as the inability to place labels in overlapping positions. These constraints are critical to consider when designing labeling systems for practical applications, including data storage.

Beyond its practical applications, recent studies shifted focus towards understanding DNA labeling from an information theoretic perspective. In~\cite{hanania2023capacity}, Hanania et al. modeled the process of DNA labeling as a communication channel. In this setup, either a fixed single label or a small set of fixed labels is being used. The input is a 4-ary sequence representing the DNA molecule and the output is a sequence in which \emph{non-zero} entries represent the presence of labels in the molecule. The work presents multiple results regarding the capacity of this channel, as a function of labels used. 
Further work studied the minimum number of labels required to achieve full labeling capacity through extremal de Bruijn subgraphs~\cite{hofmeister2026achieving}. The single-label capacity was characterized in~\cite{wu2026capacity}, while composite labels were considered in~\cite{hanania2026composite}. Error-correcting codes for labeled DNA sequences were developed for single substitution, insertion, and deletion errors~\cite{hanania2025error}, and for burst insertions and deletions~\cite{liu2026burst}.
Another related communication channel, which was studied by Nogin et al.~\cite{Nogin2023OM}, considered the process of Optical Genome Mapping (OGM)~\cite{levy2013beyond, muller2017optical}, a useful application of DNA labeling that involves optically imaging DNA fragments containing labeled short sequence patterns. The study utilizes techniques from information theory to enable the design of optimal labeling patterns for this application. In \cite{tabatabaei2020dna}, the authors introduced and investigated concepts related to using detectable chemical modifications to implement information storage leveraging native DNA. More recently, Zhang et al.~\cite{zhang2024parallel} demonstrated parallel data storage by writing epigenetic bits onto universal DNA templates using enzymatic methylation.

In enzymatic DNA labeling, the process typically involves the following steps: (1) identification of target sites on the DNA molecule, (2) binding of the labeling enzyme (e.g., CRISPR-Cas9) to these sites, and (3) attachment of the marker or tag to the DNA. The efficiency and accuracy of this process depend on factors such as the spacing between target sites and the physical properties of the enzyme. For instance, if two target sites overlap, the enzyme may not be able to bind to both sites simultaneously, leading to incomplete labeling. These practical considerations motivate the development of theoretical models that account for non-overlapping constraints, ensuring that the labeling process is both feasible and optimal.

 In this work, we define a formal framework for modeling labeling in DNA strings and investigate optimal design in such systems for the purpose of data storage in DNA. 
 In contrast to the model suggested in~\cite{hanania2023capacity}, in our model a specific known DNA molecule is considered as template for labeling, while the information is represented by patterns induced by a set of labels designed for this purpose.
Operationally, the encoder and decoder first agree on a reference sequence $S$ and a labeling code $\cC=\{\Lambda_1,\ldots,\Lambda_M\}$. Each possible message is assigned to one codeset in $\cC$. To store message $i$, the labels contained in $\Lambda_i$ are applied to the reference molecule. Every occurrence in $S$ of a selected label is marked, and the union of all marked positions forms the binary pattern $\Lambda_i(S)$. The readout records this pattern, with a \emph{one} at every labeled position and a \emph{zero} at every unlabeled position. Since both $S$ and $\cC$ are known to the decoder, it identifies the codeset that produced the observed pattern and recovers the corresponding message. Thus, the information is represented by the choice of labels applied to the known reference sequence, rather than by changing the nucleotide sequence itself. For example, if $\cC$ contains four codesets, they may
be assigned to the messages $00$, $01$, $10$, and $11$.

To address the practical limitations of enzymatic labeling, we also introduce two types of non-overlapping constraints: (1) non-overlapping, where no two label sites overlap, and (2) partially-non-overlapping, where for each label, there is at least one occurrence in which it does not overlap with any other label. The motivation for the second constraint is that having at least one isolated site for each label is sufficient to determine the label that was used, as the enzyme will bind to this site even if other occurrences are blocked. 

 We characterize the performance of this labeling channel addressing several possible variants. In particular, we consider labels of fixed length $\ell$, which is the expected case if CRISPR-Cas9 were to be used, as well as systems that can support a dynamic range of lengths. We developed, for example, an efficient encoder and decoder pair for this channel in the special case in which $\ell$ is fixed and the template sequence is $\ell$-\emph{repeat-free}. In this case, we proved optimality of our construction. 

The rest of the paper is organized as follows. In \autoref{sec:background}, we introduce the definitions that will be utilized throughout the paper, including the formal definitions of our labeling model and the problems investigated.  The codewords of our codes are codesets, i.e., each codeword is a collection of a several labels. \autoref{sec: problem 1} provides several bounds that consider strings with a small period as templates. The size of the codes associated with such periodic reference sequences is analyzed.
Moving on to \autoref{sec: fix length}, our focus shifts to labels with a fixed length, denoted as $\ell$. This section presents both upper and lower bounds on the maximum size of codes in our model. Furthermore, we introduce an efficient algorithmic construction that achieves the upper bound. Our codes, when the labels are of fixed length, are strongly associated with de Bruijn sequences and constrained codes.
Finally, \autoref{sec:conclusion} concludes the paper,  briefly summarizing some of our key findings and suggesting potential future directions.

\vspace{-0.5ex}
\section{Problem Statement}
\label{sec:background}
\vspace{-0.5ex}
Let $\Sigma_q$ denote an alphabet of size $q$.  We use $\Sigma_q^n$ to denote the set of all sequences of length $n$ over $\Sigma_q$,  and $\Sigma_q^*$ for the set of all sequences of any finite length over $\Sigma_q$. Motivated by DNA storage, in the case $q=4$ we also define $\DNA = \{\mathsf{A,C,G,T}\}$ as the DNA alphabet\footnote{Since no group structure of $\Sigma_q$ is used, $\DNA$ and $\Sigma_4$ are equivalent; this distinction is made solely for the convenience of notation.}. Accordingly, we write $\DNA^n$ and $\DNA^*$ for the sets of length-$n$ and arbitrary-length sequences over $\DNA$. Denote by $[n]$ the set $\{1, 2,\ldots, n\}$. For a symbol $\sigma \in \Sigma_q$ and an integer $\ell \geq 1$, we use $\sigma^\ell$ to denote the sequence consisting of $\ell$ consecutive copies of $\sigma$. For a sequence $S=(s_1,\ldots,s_n)\in\Sigma_q^n$, and $1\leq i \leq n-\ell+1$, let $S_{[i;\ell]}\triangleq(s_i,\ldots, s_{i+\ell-1})$ and let $$W_\ell(S)\triangleq \{S_{[i;\ell]}: 1\le i\le n-\ell+1\},$$
is the set of all length $\ell$ substrings of $S$. A \emph{label} $\lambda\in\Sigma^*_q$ is a relatively short sequence over $\Sigma_q$. For a sequence $S\in\Sigma^n_q$, the \emph{labeling} of $S$ with the label $\lambda$ is defined as follows.

\begin{definition}
Given a label $\lambda\in\Sigma^\ell_q$, the \emph{$\lambda$-labeling} (function), $\lambda:\Sigma_q^*\to\{0,1\}^*$, is the mapping (associated with the label $\lambda$) that takes a sequence $S\in\Sigma_q^n$ to the binary sequence $\lambda(S)\in \{0,1\}^n$ given by 
\[
\lambda(S)_j = 1 \iff  S_{[i;\ell]}=\lambda \text{ for some } j-\ell+1\le i\le j.
\]
That is, $\lambda(S)_{[i;\ell]}=(1, 1, \ldots, 1)$ whenever $S_{[i;\ell]}=\lambda$, and all the remaining entries of $\lambda(S)$ are \emph{zeros}.  \hfill\quad $\blacksquare$
\end{definition}

\begin{example}
   Let $\lambda={\textcolor{blue}{\mA\mA\mA\mC}}\in\DNA^4$, and consider the sequences $S_1,S_2\in\DNA^{20}$ given below.  The $\lambda$-labeling of $S_1$ and $S_2$ is as follows. \vspace{-2ex}
\[\arraycolsep=1.4pt\def\arraystretch{2.2}
\begin{array}{cccccccccccccccccccccc}
        S_1 & = &\mA&\mA& \textcolor{blue}{\mA}&\textcolor{blue}{\mA}&\textcolor{blue}{\mA}&\textcolor{blue}{\mC}& \mT&\mG&\mT&\mG&\mC&\mA&\mT&\mA& \textcolor{blue}{\mA}&\textcolor{blue}{\mA}&\textcolor{blue}{\mA}&\textcolor{blue}{\mC}& \mC& \mG \vspace{-3ex}\\
        \lambda(S_1) & = &0&0& 1&1&1&1& 0&0&0&0&0&0&0&0& 1&1&1&1& 0&0
\end{array}
\]
\vspace{-4ex}
\[\arraycolsep=1.4pt\def\arraystretch{2.2}
\begin{array}{{cccccccccccccccccccccc}}
        S_2 & = &\textcolor{blue}{\mA}&\textcolor{blue}{\mA}&\textcolor{blue}{\mA}&\textcolor{blue}{\mC}& \textcolor{blue}{\mA}&\textcolor{blue}{\mA}&\textcolor{blue}{\mA}&\textcolor{blue}{\mC}& \mT&\mC&\mC&\mG&\mG&\mA&\mT&\mG&\mT&\mG&\mG&\mA \vspace{-3ex}\\
        \lambda(S_2) &= & 1&1&1&1&1&1&1&1&  0&0&0&0&0&0&0&0&0&0&0&0
\end{array}
\]
\hfill\quad $\blacksquare$
\end{example}
\vspace{-1ex}
We will be interested in labeling using multiple labels. Let $\Lambda=\{\lambda_1,\lambda_2,\ldots,\lambda_t\}$ be a set of $t$ labels $\lambda_i\in\Sigma^*_q$. The $\Lambda$-labeling is associated with the labelings of all the $t$ labels in $\Lambda$ together and is defined as follows. 
\vspace{-0.5ex}
\begin{definition}
    Let $\Lambda = \{\lambda_1,\ldots,\lambda_t\}$ be a set of $t$ labels, where ${\lambda_i\in\Sigma^{\ell_i}_q}$.  The \emph{$\Lambda$-labeling}, ${\Lambda:\Sigma^*_q\to\{0,1\}^*}$ 
    of a sequence $S\in\Sigma^n_q$ is the bitwise OR, denoted by $\lor$, of the $t$~binary sequences ${\lambda_i}(S)$, $i\in[t]$, i.e.,
    \[
        \Lambda(S) \triangleq {\lambda_1}(S)\lor {\lambda_2}(S) \lor \cdots \lor {\lambda_t}(S).
    \]

In other words, for a set of $t$ labels $\Lambda=\{\lambda_1,\ldots,\lambda_t\}$, the $\Lambda$-labeling of a sequence $S\in\Sigma^n_q$ is the binary sequence $\Lambda(S)\in \{0,1\}^n$ in which if $S_{[i;\ell_j]}=\lambda_j\in \Lambda$ for some $j\in[t]$ then 
$\Lambda(S)_{[i;\ell_j]}=(1\ 1\ \cdots\ 1)$ and the rest are zeros.
    \hfill\quad $\blacksquare$
\end{definition}

\vspace{-0.5ex}
\begin{example}\label{ex: Lambda lebaling}
    Let $\Lambda=\{\textcolor{blue}{\mA\mA\mA\mC}, \textcolor{red}{\mC\mC}, \textcolor{teal}{\mG\mT\mG}\}$, and consider the sequences $S_1,S_2\in\DNA^{20}$ given below. 
    The $\Lambda$-labeling of $S_1$ and $S_2$ is as follow. 
    \vspace{-2ex}
\[\arraycolsep=1.4pt\def\arraystretch{2.2}
\begin{array}{cccccccccccccccccccccc}
        S_1 & = &\mA&\mA& \textcolor{blue}{\mA}&\textcolor{blue}{\mA}&\textcolor{blue}{\mA}&\textcolor{blue}{\mC}& \mT&\textcolor{teal}{\mG}&\textcolor{teal}{\mT}&\textcolor{teal}{\mG}&\mC&\mA&\mT&\mA& \textcolor{blue}{\mA}&\textcolor{blue}{\mA}&\textcolor{blue}{\mA}&\textcolor{Mulberry}{\mC}& \textcolor{red}{\mC}& \mG \vspace{-3ex}\\
        \Lambda(S_1) & = &0&0& 1&1&1&1& 0& 1&1&1& 0&0&0&0& 1&1&1&1& 1&0
\end{array}
\]
\vspace{-4ex}
\[\arraycolsep=1.4pt\def\arraystretch{2.2}
    \begin{array}{{cccccccccccccccccccccc}}
        S_2 & = &\textcolor{blue}{\mA}&\textcolor{blue}{\mA}&\textcolor{blue}{\mA}&\textcolor{blue}{\mC}& \textcolor{blue}{\mA}&\textcolor{blue}{\mA}&\textcolor{blue}{\mA}&\textcolor{blue}{\mC}& \mT&\textcolor{red}{\mC}&\textcolor{red}{\mC}&\mG&\mG&\mA&\mT&\textcolor{teal}{\mG}&\textcolor{teal}{\mT}&\textcolor{teal}{\mG}&\mG&\mA \vspace{-3ex}\\
        \Lambda(S_2) &= & 1&1&1&1&1&1&1&1&  0& 1&1& 0&0&0&0& 1&1&1& 0&0
    \end{array}
\]
    \hfill\quad $\blacksquare$
\end{example}

We note that the labeling of a sequence $S\in\Sigma_q^n$ using a label $\lambda\in\Sigma_q^*$ can be observed only if $\lambda$ is a substring of $S$, i.e., $\lambda\in W_\ell(S)$ for some integer $\ell\in[n]$. This leads us to the following definition.

\begin{definition}\label{def: redundant label}
   For $S\in\Sigma_q^n$, a label $\lambda\in\Sigma_q^\ell$ is called $S$-redundant if $\lambda\notin W_\ell(S)$. 
       \hfill\quad $\blacksquare$
\end{definition}

\begin{definition}\label{S-uniquely decodable}
    A \emph{labeling code} $\cC$ is a collection of sets of labels $\{\Lambda_i\}_{i=1}^{M}$ of sizes $\{t_i\}_{i=1}^{M}$ wherein each  $\Lambda_i\in\cC$ is called a \emph{codeset}. Given a sequence $S\in\Sigma^*_q$, we say that a labeling code $\cC$ is \emph{$S$-uniquely-decodable} if for any two distinct codesets $\Lambda_1,\Lambda_2\in \cC$, we have that $\Lambda_1(S)\ne\Lambda_2(S)$. 

    This is equivalent to saying that $\cC$ is $S$-uniquely-decodable if any codeset $\Lambda\in\cC$ can be uniquely recovered given $S$ and the $\Lambda$-labeling of $S$, $\Lambda(S)$.
        \hfill\quad $\blacksquare$
\end{definition}

Operationally, the encoder assigns each possible message to one codeset $\Lambda\in\cC$. To store that message, the labels in $\Lambda$ are applied to the reference sequence $S$, producing the binary pattern $\Lambda(S)$. The decoder knows both $S$ and $\cC$ and
recovers the message by identifying the codeset that produced the observed pattern. The requirement that $\cC$ is $S$-uniquely-decodable guarantees that this codeset is unique.

Therefore, if $|\cC|=M$, then the code can represent $M$ distinct messages, corresponding to $\log_2 M$ bits of information. More precisely, it can represent all binary messages of length $\lfloor\log_2 M\rfloor$ by using any $2^{\lfloor\log_2 M\rfloor}$ of its codesets. 

\begin{definition}
    Let $\Lambda_1,\Lambda_2$ be two sets  and let $S\in\Sigma^n_q$. We say that $\Lambda_1,\Lambda_2$ are \emph{$S$-equivalent} and denote $\Lambda_1\equiv_S\Lambda_2$ if the labeling of $S$ with $\Lambda_1$ is identical to the labeling of $S$ with $\Lambda_2$, i.e.,  $\Lambda_1(S)=\Lambda_2(S)$.
        \hfill\quad $\blacksquare$
\end{definition}

\begin{example} 
    Consider the code $\cC = \{\Lambda_{1}, \Lambda_{2},\Lambda_{3}\}$, where
    $$
    \Lambda_{1} =  \{\mA\mC\},\ \ \  \Lambda_{2}=\{\mA\mC,\mG\mT\},\ \ \  \Lambda_{3}=\{\mA, \mC\}.
    $$
    It can be verified that $\cC$ is $S$-uniquely-decodable for 
    $S = \mA\mC\mG\mT\mA\mA\mA\mT$, by observing that ${{\Lambda_{1}}(S) = 11000000}$, ${{\Lambda_{2}}(S) = 11110000}$, and ${{\Lambda_{3}}(S) = 11001110}$,
    are all different. On the other hand, $\cC$ is not $S'$-uniquely-decodable for the sequence
    $S' = \mA\mC\mG\mT\mA\mC\mG\mT$, since 
    ${\Lambda_{1}}(S') = 11001100 = {\Lambda_{3}}(S')
    $, i.e., $\Lambda_1$ and $\Lambda_3$ are $S'$-equivalent ($\Lambda_1\equiv_{S'}\Lambda_3$).
        \hfill\quad $\blacksquare$
\end{example}

As was explained in the Introduction, using common enzymatic labeling techniques, overlapping labels are not physically feasible due to the limitations of the enzymes or labeling techniques. Typically, at most one label will successfully bind, leaving the other undetectable. To address this, we introduce two more models of labeling codes:
\begin{itemize}
    \item \textbf{$S$-non-overlapping codes}
    \item  \textbf{$S$-partially-non-overlapping codes.}
\end{itemize} The former ensures that no two labels in a codeset overlap within the sequence $S$, preventing ambiguity entirely. The latter allows overlaps but guarantees that for each codeset, each label appears at least once in an isolated, non-overlapping position, ensuring that all labels can still be reliably identified.

To formally define these two models, we first name the individual sites at which a label binds. For a codeset $\Lambda$ and a sequence ${S\in\Sigma^n_q}$, let 
\[
\cO(\Lambda, S)\triangleq \set{(\lambda, i)}{\lambda\in\Lambda,\ i\in [n-|\lambda|+1],\ S_{[i;|\lambda|]}=\lambda}
\]
be the set of \emph{occurrences} of the labels of $\Lambda$ in $S$, and for an occurrence $o=(\lambda,i)\in\cO(\Lambda,S)$ let 
\[
\mathsf{supp}(o)\triangleq\{i,i+1,\ldots,i+|\lambda|-1\}\]
be the set of positions it labels. Note that $\Lambda(S)_j=1$ if and only if $j\in\mathsf{supp}(o)$ for some $o\in\cO(\Lambda,S)$.

\begin{definition}\label{def: non-overlapping}
Let $S\in\Sigma^n_q$ be a sequence. A codeset $\Lambda$ is called $S$-non-overlapping if the supports of its occurrences are pairwise disjoint, that is, $\mathsf{supp}(o)\cap\mathsf{supp}(o')=\varnothing$ for any two distinct occurrences $o,o'\in\cO(\Lambda,S)$.
A labeling code $\cC$ is called $S$-non-overlapping if every codeset $\Lambda\in\cC$ is $S$-non-overlapping.
    \hfill\quad $\blacksquare$
\end{definition}

\begin{example}\label{ex: non-overlapping}
    Let $\Lambda=\{\textcolor{blue}{\mA\mA\mA\mC}, \textcolor{red}{\mC\mC}, \textcolor{teal}{\mG\mT\mG}\}$, and consider the sequences $S_1,S_2,S_3\in\DNA^{20}$ given below. 
    \vspace{-2ex}
\[\arraycolsep=1.4pt\def\arraystretch{2.2}
\begin{array}{cccccccccccccccccccccc}
        S_1 & = &\mA&\mA& \textcolor{blue}{\mA}&\textcolor{blue}{\mA}&\textcolor{blue}{\mA}&\textcolor{blue}{\mC}& \mT&\textcolor{teal}{\mG}&\textcolor{teal}{\mT}&\textcolor{teal}{\mG}&\mC&\mA&\mT&\mA& \textcolor{blue}{\mA}&\textcolor{blue}{\mA}&\textcolor{blue}{\mA}&\textcolor{Mulberry}{\mC}& \textcolor{red}{\mC}& \mG \vspace{-3ex}\\
        \Lambda(S_1) & = &0&0& 1&1&1&1& 0& 1&1&1& 0&0&0&0& 1&1&1&1& 1&0
\end{array}
\]\vspace{-4ex}
\[\arraycolsep=1.4pt\def\arraystretch{2.2}
\begin{array}{cccccccccccccccccccccc}
        S_2 & = &\mA&\mA& \textcolor{blue}{\mA}&\textcolor{blue}{\mA}&\textcolor{blue}{\mA}&\textcolor{blue}{\mC}& \mT&\textcolor{teal}{\mG}&\textcolor{teal}{\mT}&\textcolor{teal}{\mG}& \textcolor{red}{\mC}&\textcolor{red}{\mC}&\mT&\mA& \textcolor{blue}{\mA}&\textcolor{blue}{\mA}&\textcolor{blue}{\mA}&\textcolor{Mulberry}{\mC}& \textcolor{red}{\mC}& \mG \vspace{-3ex}\\
        \Lambda(S_2) & = &0&0& 1&1&1&1& 0& 1&1&1& 1&1&0&0& 1&1&1&1& 1&0
\end{array}
\]\vspace{-4ex}
\[\arraycolsep=1.4pt\def\arraystretch{2.2}
    \begin{array}{{cccccccccccccccccccccc}}
        S_3 & = &\textcolor{blue}{\mA}&\textcolor{blue}{\mA}&\textcolor{blue}{\mA}&\textcolor{blue}{\mC}& \textcolor{blue}{\mA}&\textcolor{blue}{\mA}&\textcolor{blue}{\mA}&\textcolor{blue}{\mC}& \mT&\textcolor{red}{\mC}&\textcolor{red}{\mC}&\mG&\mG&\mA&\mT&\textcolor{teal}{\mG}&\textcolor{teal}{\mT}&\textcolor{teal}{\mG}&\mG&\mA \vspace{-3ex}\\
        \Lambda(S_3) &= & 1&1&1&1&1&1&1&1&  0& 1&1& 0&0&0&0& 1&1&1& 0&0
    \end{array}
\]
We have that, $\Lambda$ is neither $S_1$-non-overlapping nor $S_2$-non-overlapping as the label $\textcolor{red}{\mC\mC}$ overlaps with the label $\textcolor{blue}{\mA\mA\mA\mC}$ (at position $18$), which violates the condition for a non-overlapping codeset. 
As there are no overlaps of labels from $\Lambda$ in $S_3$, $\Lambda$ is an $S_3$-non-overlapping codeset.~\hfill\quad $\blacksquare$
\end{example}

\begin{definition}\label{def: partially-non-overlapping}
Let $S\in\Sigma^n_q$ be a sequence. A codeset $\Lambda$ is called $S$-partially-non-overlapping if every label $\lambda\in\Lambda$ has at least one \emph{isolated} occurrence in $S$. That is, for every $\lambda\in\Lambda$ there exists an occurrence $o=(\lambda,i)\in\cO(\Lambda,S)$ such that $\mathsf{supp}(o)\cap\mathsf{supp}(o')=\varnothing$ for every $o'\in\cO(\Lambda,S)\setminus\{o\}$. A labeling code $\cC$ is called $S$-partially-non-overlapping if every codeset $\Lambda\in\cC$ is $S$-partially-non-overlapping.
    \hfill\quad $\blacksquare$
\end{definition}

\begin{remark}
    Both definitions quantify over occurrences rather than over labels, so that a label is also required not to overlap a \emph{different occurrence of itself}. For instance, for $\Lambda=\{\mA\mA\mA\}$ and $S=\mA\mA\mA\mA\mA$ we have $\cO(\Lambda,S)=\{(\mA\mA\mA,1),(\mA\mA\mA,2),(\mA\mA\mA,3)\}$, whose supports are not pairwise disjoint, and indeed $\Lambda(S)=11111$ contains a run of \emph{ones} whose length is not a multiple of $|\lambda|$. 
        \hfill\quad $\blacksquare$
\end{remark}
\begin{example}
    Consider the codeset $\Lambda$ and the sequences $S_1,S_2,S_3$ from~\autoref{ex: non-overlapping}. As the label $\textcolor{red}{\mC\mC}$ appears only once in $S_1$, and its first symbol overlaps with an occurrence of the label $\textcolor{blue}{\mA\mA\mA\mC}$, $\Lambda$ is not $S_1$-partially-non-overlapping codeset.  By contrast, each of the three labels in $\Lambda$ has at least one isolated occurrence in $S_2$ and in $S_3$ which makes $\Lambda$ both, $S_2$-partially-non-overlapping codeset and $S_3$-partially-non-overlapping codeset.     \hfill\quad $\blacksquare$
\end{example}

Being $S$-non-overlapping or $S$-partially-non-overlapping is not necessary for a code to be $S$-uniquely-decodable: the code of \autoref{const: meet uppper bound} is $S$-uniquely-decodable, yet its codesets contain overlapping labels whenever the length of a run of \emph{ones} is not divisible by $\ell$. It is, however, sufficient in some cases, as shown in \autoref{subsec: overlap models}.

A general characterization of unique decodability beyond these specific cases remains open and motivates one of the problems stated below.

In order to represent information using DNA labeling, we need a large codebook $\cC$ which is $S$-uniquely-decodable, for some \emph{reference sequence} $S\in\DNA^n$. To this end, we need to consider both the design of such codes as well as the selection of the reference sequence~$S$. These  objectives are formally stated next, where we present three combinatorial problems related to DNA labeling and its use for information representation and data storage purposes. 

\begin{problem}[\textbf{Maximum size code for a given sequence}]\label{problem: code given S}
    Given a reference sequence $S\in\Sigma^n_q$, find the maximum size of a labeling code $\cC$ which is $S\text{-uniquely-decodable}$ and such a code $\cC$. 
    That is, we want to find the value 
        \[
        M(S)\triangleq \max\{\ |\cC|\ :\ \cC \text{ is $S$-uniquely-decodable}\},
        \]
        and an $S\text{-uniquely-decodable}$ code $\cC$ such that $|\cC| = M(S)$. 

We also consider the optimization over the reference sequence $S$ and define
\[
    M(n,q)\triangleq\max_{S\in\Sigma_q^n}M(S).
\]

     \hfill\quad $\blacksquare$
     \end{problem}

Thus, $M(S)$ is the maximum number of messages that can be represented using the reference sequence $S$, and $\log_2 M(S)$ is the corresponding maximum amount of information in bits. Similarly, $M(n,q)$ is the maximum number of messages when the reference sequence
may be selected from $\Sigma_q^n$, and the corresponding amount of information is $\log_2 M(n,q)$ bits. The information content per nucleotide/position is
\[
    \frac{\log_2 M(S)}{n}
\]
bits per reference-sequence symbol.

\full{
\begin{problem}[\textbf{Shortest reference sequence for a given code or code size}]
Given a labeling code $\cC$, find the following.
\begin{enumerate}
    \item Characterize the reference sequences    $S\in\Sigma_q^*$ for which $\cC$ is    $S$-uniquely-decodable. In particular, determine necessary and sufficient conditions for such a reference sequence to exist.

    \item If such a reference sequence exists, determine the minimum length $n$ for which $\cC$ is $S$-uniquely-decodable for some $S\in\Sigma_q^n$.

    \item Construct a reference sequence $S$ of minimum length for which $\cC$ is $S$-uniquely-decodable.
\end{enumerate}

We also consider the case in which the code itself is not prescribed and only its desired size is given. For a positive integer $M$, define
\[
N(M,q)
\triangleq
\min\left\{
n\in\mathbb N:
M(n,q)\ge M
\right\}.
\]
Determine $N(M,q)$ and construct a reference sequence $S\in\Sigma_q^{N(M,q)}$ together with an $S$-uniquely-decodable labeling code $\cC$ of size $M$.
\hfill\quad $\blacksquare$
\end{problem}
}

In the next two examples, we show that given a non-trivial labeling code $\cC$, it might be that $\cC$ is not $S$-uniquely-decodable for all $S\in\Sigma^*_q$. Otherwise, there exists a minimal length~$n$ such that $\cC$ is $S$-uniquely-decodable for some $S\in\Sigma^n_q$ but is not $S$-uniquely-decodable for any $S\in\Sigma^{n'}_q$ for $n'<n$. 

\begin{example}\label{ex: prob 2 |S|=4}
    Let $\cC = \{\Lambda_1,\Lambda_2,\Lambda_3\}$ such that ${\Lambda_1=\{\mA\}}$, $\Lambda_2=\{\mA\mA\}$, and  $\Lambda_3=\{\mA\mA\mC\}.$
    We first show that $\cC$ is not $S$-uniquely-decodable for any sequence $S$ of length at most three. 
    \begin{itemize}
        \item  If $|S|=1$, then the labels $\mA\mA$ and $\mA\mA\mC$ are $S$-redundant, and hence $\Lambda_2(S)=\Lambda_3(S)=0$.
        \item If $|S|=2$, then $\mA\mA\mC$ is $S$-redundant. If $S=\mA\mA$, then    $\Lambda_1(S)=\Lambda_2(S)=11$; otherwise, $\mA\mA$ is also $S$-redundant and $\Lambda_2(S)=\Lambda_3(S)=00$.
        \item Finally, suppose that $|S|=3$. If $\mA\mA$ is not a substring of $S$, then neither is $\mA\mA\mC$, and consequently $\Lambda_2(S)=\Lambda_3(S)=000$. Otherwise, all occurrences of $\mA$ belong to the same run of $\mA$'s of length at least two, and therefore $\Lambda_1(S)=\Lambda_2(S)$.
    \end{itemize}

    On the other hand, for
    \[
        S=\mA\mA\mC\mA
    \]
    we have
    \[
        \Lambda_1(S)=1101,\qquad
        \Lambda_2(S)=1100,\qquad
        \Lambda_3(S)=1110.
    \]
    Hence, $\cC$ is $S$-uniquely-decodable, and each of its labels occurs in $S$. Thus, the minimal length is four. \hfill\quad $\blacksquare$
\end{example}

\begin{example}\label{ex: prob 2 |S|=inf}
     Let $\cC = \{\Lambda_1,\Lambda_2,\Lambda_3\}$ such that 
     \[
\begin{aligned}
    \Lambda_1
        &=\{\mA,\mC,\mG,\mT\},\\
    \Lambda_2
        &=\{\mA\mA,\mA\mC,\mA\mG,\mA\mT,
             \mC\mA,\mC\mC,\mC\mG,\mC\mT,\\
        &\hspace{1.7em}
             \mG\mA,\mG\mC,\mG\mG,\mG\mT,
             \mT\mA,\mT\mC,\mT\mG,\mT\mT\},\\
    \Lambda_3
        &=\{\mA\mA\mA\}.
\end{aligned}
\]
     If $|S|=1$, then
    $\Lambda_2(S)=\Lambda_3(S)=0$. If $|S|\geq2$, every
    position of $S$ is covered by both a label from $\Lambda_1$
    and a label from $\Lambda_2$, and hence
    \[
        \Lambda_1(S)=\Lambda_2(S)=1^{|S|}.
    \]
Therefore, $\cC$ is not $S$-uniquely-decodable for any
${S\in\{\mA,\mC,\mG,\mT\}^*}$.

\hfill\quad $\blacksquare$
\end{example}

In many labeling techniques, not all labels $\lambda \in \Sigma^*_q$ are valid, as certain constraints must be considered. For example, when using CRISPR-Cas9, valid labels must be exactly 23 nucleotides long, while for CRISPR-Cas12, they are slightly longer (27 or 28 nucleotides). Furthermore, not all sequences of length $\ell$ are permissible, as the gRNA recognizes a specific suffix known as the \textit{protospacer adjacent motif (PAM)}, required for binding. For instance, in the case of CRISPR-Cas9, a valid label must be 23 nucleotides long, with its final three bases forming part of a ${\mN\mG\mG}$ suffix (where $\mN$ can be any nucleotide).

This motivates a generalization of the previous problems to \autoref{problem: executable labels}, which introduces a variation of the model that accounts for compatible labels. For example, in the case of CRISPR-compatible labels~\cite{jinek2012programmable, chen2018efficient, cong2013multiplex, amit2021crispector} each $\lambda_i$ can be thought of as a potential edit site.

\begin{problem}[\textbf{A restricted set of labels}]\label{problem: executable labels}
    Given a collection of $\tau$ labels ${\cV = \{\lambda_1,\lambda_2,\ldots,\lambda_\tau\}}$, referred as the \emph{executable labels}, find the following.
    \begin{enumerate}
    \item \textbf{(Maximum size code for a given sequence).} Given a reference sequence $S\in\Sigma^n_q$, find the value 
    \[
     M(S,\cV)  \triangleq  \max\left\{ |\cC|\ \middle\vert \begin{array}{l}
    \text{$\cC$ is $S$-uniquely-decodable,} \\
    \ \ \ \ \ \ \ \ \ \ \cC\subseteq \cP(\cV)
  \end{array}\right\},
    \]
    where $\cP(\cV)$ is the power set of $\cV$ (i.e., for any $\Lambda\in\cC$ we have $\Lambda\subseteq \cV$), and an $S$-uniquely-decodable code $\cC$ such that $|\cC|=M(S,\cV)$. 
    \item \textbf{(Maximum size code and a corresponding sequence for a given length).} Given a positive integer $n$ find the value $$M(n, q, \cV)\triangleq\max_{S\in\Sigma^n_q}\left\{ M(S,\cV)\right\},$$ a reference sequence $S_{(n,\cV)}\in\Sigma^n_q$ such that 
    \[
    M(S_{(n,\cV)},\cV) = M(n, q, \cV),
    \]
    and an $S_{(n,\cV)}$-uniquely decodable code $\cC_{(n,\cV)}\subseteq\cP(\cV)$, such that $|\cC_{(n,\cV)}|= M(n, q, \cV)$. 
 \item \textbf{(Maximum size code).} Find the size of a maximum code $\cC\subseteq \cP(\cV)$ that is $S$-uniquely-decodable for some sequence $S\in\Sigma^*_q$, i.e., find the value 
 \[
 M(q, \cV)\triangleq \max_{n\in\N} \{M(n, q, \cV)\}.
 \]
 Additionally, find a shortest sequence $S\in\Sigma_q^*$ and a code ${\cC\subseteq \cP(\cV)}$ of size $M(q, \cV)$ such that $\cC$ is $S$-uniquely-decodable.
    \end{enumerate}
        \hfill\quad $\blacksquare$
\end{problem}

The three problems defined above will be discussed under the three different models we defined for $S$-uniquely-decodable codes:
\begin{itemize}
    \item Model 1: $S$-uniquely-decodable codes with no additional requirements.
    \item Model 2: Non-overlapping $S$-uniquely-decodable codes.
    \item Model 3: Partially non-overlapping $S$-uniquely-decodable codes.
\end{itemize}

For $\mathsf{X}\in\{\mathsf{NO},\mathsf{PNO}\}$, let
$M^{\mathsf{X}}(S,\cV)$ denote the maximum size of an
$S$-uniquely-decodable code $\cC\subseteq\cP(\cV)$ that is,
respectively, $S$-non-overlapping or
$S$-partially-non-overlapping. We similarly define
\[
M^{\mathsf{X}}(n,q,\cV)
    \triangleq
    \max_{S\in\Sigma_q^n}M^{\mathsf{X}}(S,\cV)
\]
and
\[
M^{\mathsf{X}}(q,\cV)
    \triangleq
    \max_{n\in\mathbb N}M^{\mathsf{X}}(n,q,\cV).
\]
When the set of executable labels is unrestricted, we omit $\cV$
from the notation.

\begin{observation}
    Let $\cC$ be a labeling code. If $\cC$ is a non-overlapping $S$-uniquely-decodable code then $\cC$ is also a partially non-overlapping $S$-uniquely-decodable code. If $\cC$ is a partially non-overlapping $S$-uniquely-decodable code then $\cC$ is an $S$-uniquely-decodable code (with no additional requirements). 
        \hfill\quad $\blacksquare$
\end{observation}

\begin{corollary}\label{cor: overlapping relations}
    For any $S\in\Sigma_q^n$, we have that 
    $$
    M(S) \ge M^{\mathsf{PNO}}(S) \ge M^{\mathsf{NO}}(S).
    $$
    Moreover, for any set of executable labels $\cV$, any integers $n,q>1$ and any $S\in\Sigma_q^n$ we have that
    \begin{align*}
        M(S,\cV) &\ge M^{\mathsf{PNO}}(S,\cV) \ge M^{\mathsf{NO}}(S, \cV),\\
        M(n,q,\cV) &\ge M^{\mathsf{PNO}}(n,q,\cV) \ge M^{\mathsf{NO}}(n,q, \cV),\\
        M(q,\cV) &\ge M^{\mathsf{PNO}}(q,\cV) \ge M^{\mathsf{NO}}(q, \cV).
    \end{align*}
     \hfill\quad $\blacksquare$
\end{corollary}

\section{Basic Results using Periodicity}
\label{sec: problem 1}

We first determine when the trivial upper bound $M(n,q)\le 2^n$ is attained. We then study $M(S)$ for periodic reference sequences, starting with sequences of period one and proceeding to larger periods.

\begin{observation}\label{obs: maximum general code}
For any positive integers $n$ and $q>1$,
\[
    M(n,q)=2^n
    \quad\Longleftrightarrow\quad
    n\le q.
\]
\end{observation}

\begin{IEEEproof}
Since every labeling of a sequence of length $n$ is a binary sequence of length $n$, we have $M(n,q)\le 2^n$. If $n\le q$, choose $S=s_1s_2\cdots s_n$ with $s_1,\ldots,s_n$ distinct. Every binary sequence can then be obtained by selecting the corresponding subset of the length-one labels $\{s_1,\ldots,s_n\}$, and hence $M(S)=2^n$.

Conversely, suppose that $M(n,q)=2^n$. Then there exists $S\in\Sigma_q^n$ for which every binary sequence of length $n$ is a labeling of $S$. In particular, for every $i\in[n]$, the sequence $0^{i-1}10^{n-i}$ must be a labeling of $S$. Such a labeling can only be obtained using a length-one label whose symbol occurs only at position $i$. Thus, all the symbols of $S$ are distinct, and therefore $n\le q$.
\end{IEEEproof}

We begin the analysis of periodic sequences with a sequence of period one, i.e., the case $\pi(S)=1$. That is, $S=\sigma^n$ for some $\sigma\in\Sigma_q$ and some positive integer~$n$.

\begin{lemma}\label{lem: S for minimal code}
    If $S\in\Sigma^n_q$ consists of a single run of a symbol $\sigma\in\Sigma_q$, then $M(S)=2$, and this value is attained by the code $\cC_\sigma=\{\varnothing, \{\sigma\}\}$.
\end{lemma} 

\begin{IEEEproof}
To see that $\cC_\sigma$ is $S$-uniquely-decodable note that for $\Lambda_1=\varnothing$ and $\Lambda_2=\{\sigma\}$ we have that $\Lambda_1(S)$ is the \emph{all-zero} word while $\Lambda_2(S)$ is the \emph{all-one} word.

Let $\cC$ be an $S$-uniquely-decodable labeling code with maximum size. Note that for any label $\lambda\in\Sigma^*_q$, if $\lambda\ne \sigma^i$, for any integer $0< i\le |S|$, then $\lambda$ is not a substring of $S$, i.e., $\lambda$ is $S$-redundant and $\lambda(S)$ is the \emph{all-zero} word. Otherwise, we have that $\lambda(S)$ is the \emph{all-one} word. By definition of the bitwise OR operation, for any codeset $\Lambda\in \cC$ we have that either $\Lambda\equiv_S \varnothing$  or $\Lambda\equiv_S\{\sigma\}$. Since $\cC$ is  $S$-uniquely-decodable, the latter implies that $|\cC|\le 2$.
\end{IEEEproof}

To extend \autoref{lem: S for minimal code} to more involved cases, we first present the definitions of a \emph{period}.

\begin{definition}
    For $S \in \Sigma^n_q$, an integer $1 \leq \pi \leq n-1$ is called a \emph{period} of $S$ if $\pi\mid n$ and ${S_{i} = S_{i+\pi}}$ for all $1 \leq i \leq n-\pi$. Additionally, let $\pi(S)$ be the least period of the sequence $S$ if such a period exists, and $\pi(S)=n$ otherwise.         \hfill\quad $\blacksquare$
\end{definition}

\begin{lemma} \label{lem: period two}
If $S\in\Sigma^n_q$, where $q>1, n>2$ and $\pi(S) = 2$, then 
$M(S) = 7$. 
\end{lemma}
\begin{IEEEproof}
    Assume w.l.o.g. that $S=\mA\mC\mA\mC\ldots \mA\mC$.
     The only labels that can be considered are the ones that are substrings of $S$, i.e., the possible labels are
     \vspace{-1.1ex}
     \[
\bigcup_{t=0}^{\frac{n}{2}-1}
\left\{
(\mA\mC)^{t+1},
(\mA\mC)^t\mA,
(\mC\mA)^t\mC
\right\}
\ \cup\
\bigcup_{t=1}^{\frac{n}{2}-1}
\left\{
(\mC\mA)^t
\right\}.
\]
     Furthermore, for any ${1\le t\le {n}/{2}}$ we have that 
         ${\{\mA\mC\}\equiv_S \{(\mA\mC)^t\}}$, 
     and for any ${1\le t\le {n}/{2}-1}$ we have that 
           ${\{\mC\mA\}\equiv_S \{(\mC\mA)^t\}}$,
         $\{\mA\mC\mA\}\equiv_S \{(\mA\mC)^t\mA\}$, and
         ${\{\mC\mA\mC\}\equiv_S \{(\mC\mA)^t\mC\}}$.
     Thus, it is sufficient to consider only the labels in $\{\mA, \mC, \mA\mC, \mC\mA, \mA\mC\mA, \mC\mA\mC\}$. It can be verified that the labeling of $S$ by each of these labels is unique, and by considering the empty codeset we have an $S$-uniquely decodable code of size $7$. To see that no larger code is $S$-uniquely decodable, note that any possible codeset is $S \text{-equivalent}$ to some codeset that is composed only of labels in $\{\mA, \mC, \mA\mC, \mC\mA, \mA\mC\mA, \mC\mA\mC\}$, or to the empty codeset, and that any codeset that is a subset of the these six labels, is $S$-equivalent to a codeset that consists of a single label out of these six labels.
\end{IEEEproof}

Next, we generalize the result in \autoref{lem: period two} to any period $\pi$, in the cases where the least period size, $\pi(S)$, is at most the alphabet size. The proof, which can be found in the appendix, is done by considering all possible cases and counting the contribution of each. 
\begin{lemma}\label{lem: period bound}
    Let $q>1$ and let $S\in\Sigma^n_q$ with $\pi\triangleq\pi(S)\le \min\{n,q\}$. Then 
    $$M(S)\le \begin{cases}
        2^\pi & \text{ if } n=\pi,\\
        2^{\pi-3}(\pi^2+3\pi+4) & \text{ if } n = 2\pi\\
      {2^{\pi-2}(\pi^2+\pi+1),}
            &\text{ if } n\ge3\pi.
    \end{cases}$$
   Moreover, equality holds whenever $S_{[1;\pi]}$ consists of
    $\pi$ distinct symbols.
\end{lemma}

\section{Fixed-Length Labels}\label{sec: fix length}
In this section, we consider the special case in which all the labels have the same length. That is, we are interested in $S$-uniquely-decodable labeling codes (for some sequence $S\in\Sigma^*_q$), in which all the labels that are being used are of a specific length.

\begin{definition}
    A labeling code $\cC$ is called an \emph{$\ell$-labeling code} if for any label $\lambda$ in any codeset $\Lambda\in\cC$ we have that $|\lambda|=\ell$. 
    
    \hfill\quad $\blacksquare$
\end{definition} 

We consider \autoref{problem: executable labels} (a restricted set of labels) in the special case where the executable labels are all the labels of length $\ell$, for some integer $\ell$, i.e., $\cV=\Sigma_q^\ell$. To this end, for a reference sequence $S$ and an integer $n$ we define $M_\ell(S)\triangleq M(S,\Sigma^\ell_q)$ and $M_\ell(n, q) \triangleq M(n,q,\Sigma^\ell_q)$, where $M(S,\Sigma^\ell_q)$ is the maximal size of an $\ell$-labeling $S$-uniquely-decodable code, for a given $S\in\Sigma_q^n$, and $ M(n,q,\Sigma^\ell_q)$ is the maximal value of $M(S,\Sigma^\ell_q)$, when considering any $S\in\Sigma_q^n$ (see~\autoref{problem: executable labels}). 
Accordingly, $\log_2 M_\ell(S)$ is the maximum number of bits that can be represented on $S$ using labels of length $\ell$, while $\log_2 M_\ell(n,q)$ is the corresponding maximum when the reference sequence may also be selected.

For $\mathsf{X}\in\{\mathsf{NO},\mathsf{PNO}\}$, we define
\[
M_\ell^{\mathsf{X}}(S)
    \triangleq M^{\mathsf{X}}(S,\Sigma_q^\ell)
\quad\text{and}\quad
M_\ell^{\mathsf{X}}(n,q)
    \triangleq M^{\mathsf{X}}(n,q,\Sigma_q^\ell).
\]
Similarly, we use the notations $S_{(n,\ell)}$ and $\cC_{(n,\ell)}$ for the related reference sequence and code.

We start with an upper bound on $M_\ell(n, q)$ for any integers $q>1$ and $0<\ell\le n$. Then, we show that in the special case where $\ell\ge\log_q(n-\ell+1)$, this bound is tight. 
Before we present this bound and an explicit construction that meets it for $\ell\ge\log_q(n-\ell+1)$, let us define  $\cR_\ell$ to be the set of all binary sequences in which any run of \emph{ones} is of length at least~$\ell$. Additionally, let  $\cR_\ell(n)\triangleq \cR_\ell\cap\{0,1\}^n$ be the set of length $n$ in $\cR_\ell$. 

The constraint $\cR_\ell$ is strongly related to the well-known \emph{run-length limited} (\emph{RLL}) constraints; see, for example,~\cite{marcus2001introduction} for a comprehensive survey. We recall the RLL definition before explaining its difference from the $R_\ell$ constraint.

\begin{definition}
A binary sequence satisfies the \emph{$(d,k)$-RLL constraint} if between every two consecutive \emph{ones}, there are at least $d$ and at most $k$ \emph{zeros}.
\hfill\quad $\blacksquare$
\end{definition}

Exchanging the roles of \emph{zeros} and \emph{ones} in the $(\ell,\infty)$-RLL constraint would require at least $\ell$ \emph{ones} between every two successive occurrences of \emph{zero}. In particular, such a constraint forbids two adjacent \emph{zeros}. This is different from $\cR_\ell$, which allows runs of \emph{zeros} of arbitrary length. More precisely, between two successive \emph{zeros}, the number of \emph{ones} in a sequence from $\cR_\ell$ may be either zero or at least $\ell$. Furthermore, the same minimum-length requirement applies to runs of \emph{ones} at the beginning or end of the sequence. Thus, $\cR_\ell$ is not equivalent to the $(\ell,\infty)$-RLL constraint, even after exchanging the roles of \emph{zeros} and \emph{ones}. Nevertheless, it is what is called a ``regular constrained system", and its capacity can be determined from a finite-state presentation of the constrained system~\cite{marcus2001introduction}.

\begin{theorem}\label{th: upper bound}
Let $\ell,n$, and $q>1$ be positive integers. For every sequence $S\in\Sigma_q^n$,
\[
M_\ell^{\mathsf{NO}}(S)
\le
M_\ell^{\mathsf{PNO}}(S)
\le
M_\ell(S)
\le
|\cR_\ell(n)|.
\]
\end{theorem}

\begin{IEEEproof}
The first two inequalities follow directly from the definitions of the three models: every $S$-non-overlapping code is $S$-partially-non-overlapping, and every $S$-partially-non-overlapping code is an unrestricted $S$-uniquely-decodable code.

Let $\cC$ be an $S$-uniquely-decodable $\ell$-labeling code. Note that by the definition of $\ell$-labeling code, for any $\Lambda \in \cC$, we have that $\Lambda(S)$ is a binary string of length $|S|$, in which any run of \emph{ones} is of length $\ell$ or more. Additionally, since $\cC$ is $S$-uniquely-decodable, for any  $\Lambda_1, \Lambda_2 \in \cC$, we have that $\Lambda_1(S)\ne \Lambda_2(S)$, and thus,  $|\cC|\le  |\cR_\ell(n)|$.
\end{IEEEproof}

\begin{corollary}\label{cor: upper bound}
    For any three positive integers $\ell, n$ and $q>1$ we have that $M_\ell(n,q)\le |\cR_\ell(n)|$.  \hfill\quad $\blacksquare$
\end{corollary}

In \autoref{th: code for l-repeat-free} and \autoref{lem: not repeat-free} we show that the last inequality in \autoref{th: upper bound} is tight if and only if $S$ is \emph{$\ell$-repeat-free} sequence~\cite{elishco2021repeat}, i.e., any substring of length $\ell$ in $S$ is unique.

\begin{definition}\label{def: repeat-free}
    A sequence $S\in\Sigma_q^n$ is $\ell$-repeat-free for some integer $\ell\in[n]$ if for any two integers $1\le i<j\le n-\ell+1$ we have that $S_{[i;\ell]}\ne S_{[j;\ell]}$.
        \hfill\quad $\blacksquare$
\end{definition}
 For the non-overlapping model, the bound can be strengthened since every run of \emph{ones} in the induced labeling must have length divisible by $\ell$. Moreover, when $S$ is $\ell$-repeat-free, the non-overlapping and partially-non-overlapping models coincide, since every nonredundant label has exactly one occurrence in $S$. We return to these two models in \autoref{subsec: overlap models}.

We next determine the asymptotic size of $\cR_\ell(n)$ when $\ell$ is fixed.

\begin{lemma}\label{lem: R ell capacity}
Let $\ell$ be a fixed positive integer, and let $\rho_\ell\in(1,2]$ be the unique real solution in this interval of
\[
    \rho_\ell^{\ell-1}(\rho_\ell-1)^2=1.
\]
Then
\[
\mathsf{cap}(\cR_\ell)
\triangleq
\lim_{n\to\infty}
\frac{\log_2|\cR_\ell(n)|}{n}
=
\log_2\rho_\ell.
\]
\end{lemma}

\begin{IEEEproof}
Consider the finite-state automaton with states $v_0,v_1,\ldots,v_\ell$. The state $v_0$ represents the empty sequence or a sequence ending in a \emph{zero}; the state $v_i$, for $1\le i<\ell$, represents a sequence ending in a run of exactly $i$ \emph{ones}; and $v_\ell$ represents a sequence ending in a run of at least $\ell$ \emph{ones}. The initial state is $v_0$, and the accepting states are $v_0$ and $v_\ell$. The automaton is illustrated below.

\begin{center}
\begin{tikzpicture}[
    >=stealth,
    shorten >=1pt,
    node distance=1.25cm,
    on grid,
    auto,
    initial text={},
    every state/.style={
        minimum size=7mm,
        inner sep=1pt
    }
]
\node[state,initial,accepting] (v0) {$v_0$};
\node[state] (v1) [right=of v0] {$v_1$};
\node (dots) [right=of v1] {$\cdots$};
\node[state] (vlm) [right=of dots] {$v_{\ell-1}$};
\node[state,accepting] (vl) [right=of vlm] {$v_\ell$};

\path[->]
(v0)  edge[loop above] node {$0$} (v0)
      edge              node {$1$} (v1)
(v1)  edge              node {$1$} (dots)
(dots) edge              node {$1$} (vlm)
(vlm) edge              node {$1$} (vl)
(vl)  edge[loop above] node {$1$} (vl)
      edge[bend left=42] node[below] {$0$} (v0);
\end{tikzpicture}
\end{center}

Every transition not shown from one of the intermediate states leads to a rejecting sink state, which is omitted from the diagram.

For a directed graph with ordered vertices $v_0,v_1,\ldots,v_\ell$, its \emph{adjacency matrix} is the matrix $A_\ell=(a_{i,j})$, where $a_{i,j}$ is the number of directed edges from $v_i$ to $v_j$. Thus, the adjacency matrix of the displayed graph is
\[
A_\ell=
\begin{pmatrix}
1&1&0&\cdots&0\\
0&0&1&\cdots&0\\
\vdots&&\ddots&\ddots&\vdots\\
0&\cdots&0&0&1\\
1&0&\cdots&0&1
\end{pmatrix}.
\]

A nonnegative square matrix $A$ is called \emph{primitive} if there exists a positive integer $m$ such that every entry of $A^m$ is positive. Recall that the Perron--Frobenius~\cite[Chapter 3]{marcus2001introduction} theorem states that if $A$ is primitive, then its spectral radius $\rho(A)$ is a positive eigenvalue that is strictly larger in absolute value than all its other eigenvalues. Moreover, the exponential growth rate of the entries of $A^n$ is $\rho(A)$.

The displayed graph is strongly connected and contains a self-loop, and hence $A_\ell$ is primitive. Since the automaton is deterministic, the number of accepted sequences of length $n$ is
\[
|\cR_\ell(n)|
=
(A_\ell^n)_{0,0}+(A_\ell^n)_{0,\ell}.
\]
It therefore follows from the Perron--Frobenius theorem that
\[
\mathsf{cap}(\cR_\ell)
=
\lim_{n\to\infty}
\frac{\log_2|\cR_\ell(n)|}{n}
=
\log_2\rho(A_\ell).
\]

It remains to determine $\rho(A_\ell)$. Recall that the eigenvalues of a matrix $A_\ell$ are the roots of $\det(xI-A_\ell)$. A direct calculation gives
\[
\det(xI-A_\ell)
=
x^{\ell-1}(x-1)^2-1.
\]
Hence, $\rho(A_\ell)$ satisfies
\[
\rho(A_\ell)^{\ell-1}
\bigl(\rho(A_\ell)-1\bigr)^2
=
1.
\]

The function $x^{\ell-1}(x-1)^2$ is strictly increasing from zero to infinity on $(1,\infty)$, so this equation has a unique solution greater than one. Since its value at $x=2$ is at least one, this solution belongs to $(1,2]$ and is precisely $\rho_\ell$. Therefore,
\[
\mathsf{cap}(\cR_\ell)=\log_2\rho_\ell,
\]
as claimed.
\end{IEEEproof}

\begin{corollary}\label{cor: RLL capacity}
For every fixed positive integer $\ell$ and every ${q>1}$,
\[
    \limsup_{n\to\infty}
    \frac{\log_2 M_\ell(n,q)}{n}
    \le
    \log_2\rho_\ell,
\]
where $\rho_\ell\in(1,2]$ is the unique real solution in this interval of
\[
    \rho_\ell^{\ell-1}(\rho_\ell-1)^2=1.
\]
\end{corollary}

\begin{IEEEproof}
The result follows immediately from \autoref{cor: upper bound}
and \autoref{lem: R ell capacity}.
\end{IEEEproof}

Next, we show that the bound in \autoref{cor: upper bound} is tight by presenting an explicit construction for a code $\cC_{(n,\ell)}$ that meets this bound with equality when $S$ is an \emph{$\ell$-repeat-free} sequence.

\begin{definition}\label{def: de~Bruijn}
    A de~Bruijn sequence of order $\ell$ is a cyclic sequence of length $q^\ell$  in which every possible length-$\ell$ substring appears exactly once.  If we unwrap such a cyclic sequence into an acyclic one (by appending the  first $\ell-1$ symbols to the end), we obtain a sequence of length  $q^\ell+\ell-1$ in which all length-$\ell$ substrings are distinct. Such a sequence exists for any alphabet of size $q$ and any positive integer $\ell$~\cite{van1987circuits, etzion2024sequences}. \hfill\quad $\blacksquare$
\end{definition}
\begin{remark}\label{remark: de~Bruijn}
     An $\ell$-repeat-free sequence of length $n\le q^\ell+\ell-1$ over $\Sigma_q$ can be obtained by taking any length-$n$ substring of the latter. Hence, an $\ell$-repeat-free sequence of length $n$ exists if and only if  $n \le q^\ell+\ell-1$, or equivalently, if and only if  $\ell \ge \log_q (n-\ell+1)$.     \hfill\quad $\blacksquare$
\end{remark}

\begin{construction}\label{const: meet uppper bound}
    Let $S$ be an $\ell$-repeat-free sequence of length~$n$ over $\Sigma_q$. 
    For any $X=(x_1,\ldots, x_n)\in \cR_\ell(n)$, we define the codeset $\Lambda_X$ as follows. 
    \begin{enumerate}
        \item Initialize an empty set $\Lambda_X =\varnothing$. 
        \item For any run of \emph{ones} in $X$ of length $\ell'\geq \ell$, $X_{[i;\ell']}$:
        \begin{enumerate}
            \item Add the labels $S_{[i;\ell]},S_{[i+\ell;\ell]},\ldots,S_{\left[i+\ell\left(
\left\lfloor{\ell'}/{\ell}\right\rfloor-1
\right);\ell
\right]}
$ to $\Lambda_X$.
            \item  If  $\ell'/\ell$ is not an integer, add the label $S_{[i+\ell'-\ell;\ell]}$ into $\Lambda_X$.
        \end{enumerate}
    \end{enumerate}  
    That is, for any $X\in\cR_\ell(n)$, the codeset $\Lambda_X$ is composed of labels that corresponds to the locations in which $X$ has runs of \emph{ones}, where for each run of \emph{ones} we consider all the length-$\ell$ consecutive non-overlapping labels (in the case where $\ell$ doesn't divide the length of the run, we also add the label that corresponds to the length-$\ell$ suffix of that run).   
    Finally, we define the code 
    $$\cC_{(n,\ell)} = \{\Lambda_X: X\in \cR_\ell(n)\}.$$
        \hfill\quad $\blacksquare$
\end{construction}

\begin{example}
Let $S\in\DNA^{25}$, and $X\in\cR_4(25)$ be the following sequences, 
\vspace{-2ex}
\[
\arraycolsep=0.8pt\def\arraystretch{2.0}
\begin{array}{{ccccccccccccccccccccccccccc}}
S & = & \mA& \mA& \mA& \mA& \mC& \mA& \mA& \mA& \mG& \mA& \mA& \mA& \mT& \mA& \mA& \mC& \mC& \mA& \mA& \mC& \mG& \mA& \mA& \mC& \mT, \vspace{-1.5ex}\\
X & = & 1& 1& 1& 1& 1& 1& 1& 1& 1& 1& 1& 0& 0& 0& 0& 0& 1& 1& 1& 1& 0& 0& 0& 0& 0.
\end{array}
\vspace{-1ex}
\]
The sequence $S$ is a $4$-repeat-free sequence, and the codeset $\Lambda_X$ which is obtained by \autoref{const: meet uppper bound} for the binary sequence $X$ is
$\Lambda_X = \left\{ 
\mA\mA\mA\mA, \ \mC\mA\mA\mA, \ \mA\mG\mA\mA, \ \mC\mA\mA\mC
\right\}$,
where the first three labels correspond to the first run of \emph{ones} (in $X$) and the last label corresponds to the second run of \emph{ones}. 
    \hfill\quad $\blacksquare$
\end{example}
\begin{theorem}\label{th: code for l-repeat-free}
    For any $\ell$-repeat-free  sequence $S\in\Sigma^n_q$, the code $\cC_{(n,\ell)}$ obtained by \autoref{const: meet uppper bound} is an $S$-uniquely-decodable $\ell$-labeling code of size $|\cR_\ell(n)|$. 
\end{theorem}
\begin{IEEEproof}
     For any $\Lambda\in\cC_{(n,\ell)}$, any label $\lambda\in\Lambda$ is of the form $S_{[j;\ell]}$ for some $j$, and by definition we have that all the labels are of length exactly $\ell$. 
     That is, $\cC_{(n,\ell)}$ is an $\ell$-labeling code. Furthermore, for any $X\in\cR_\ell(n)$, the labels selected by the construction cover exactly the runs of \emph{ones} in $X$. Since $S$ is $\ell$-repeat-free, each selected label occurs in $S$ only at the position from which it was selected. Hence, $\Lambda_X(S)=X$. Thus, for any two distinct sequences $X,Y\in \cR_\ell(n)$,
     we have that $\Lambda_X(S)=X\ne Y=\Lambda_Y(S)$ which implies that $\cC_{(n,\ell)}$ is an $S$-uniquely-decodable code of size $|\cR_\ell(n)|$. 
\end{IEEEproof}

\begin{remark}
    We note that to encode a binary message using this construction, we first need to encode it into a sequence $X\in\cR_\ell(n)$ and then apply \autoref{const: meet uppper bound} on $X$. An efficient encoder of the $\cR_\ell$ constraint can be achieved by constructing a deterministic finite automaton that accepts all the constrained words, and then utilizing the algorithm in \cite{adler1983algorithms}, known as the state-splitting algorithm, to design encoders with rate approaching the capacity. 
            \hfill\quad $\blacksquare$
    \end{remark}

The following lemma shows that the size of an $\ell$-labeling code is maximized for $S$ of length $n$, which is $\ell$-repeat-free (if such $S$ exists). 
\begin{lemma}\label{lem: not repeat-free}
    If $S\in\Sigma^n_q$ is not $\ell$-repeat-free, then $$M_\ell(S)<|\cR_\ell(n)|.$$ 
\end{lemma}
\begin{IEEEproof}
    By the proof of~\autoref{th: upper bound}, for any codeset $\Lambda$, we have that $\Lambda(S)\in \cR_\ell(n)$.  Additionally, 
    since $S\in\Sigma^n_q$ is not $\ell$-repeat-free, there exist two indices $1\le i<j\le n-\ell+1$, such that $S_{[i;\ell]}=S_{[j;\ell]}$. Hence, assume $\Lambda(S) = 0^{i-1}1^\ell0^{n-\ell-i+1}\in\cR_\ell(n)$. It implies that there exists a label $\lambda\in \Lambda$ such that $\lambda = S_{[i;\ell]}$ and therefore 
    $\Lambda(S)_{[i;\ell]} = \Lambda(S)_{[j;\ell]} = 1^\ell$, i.e., 
    $\Lambda(S) \ne 0^{i-1}1^\ell0^{n-\ell-i+1}$, a contradiction. Since $ 0^{i-1}1^\ell0^{n-\ell-i+1}\in \cR_\ell(n)$ it implies that $M_\ell(S)<|\cR_\ell(n)|.$
    \end{IEEEproof}

 \autoref{th: upper bound}, \autoref{cor: upper bound} and \autoref{lem: not repeat-free} state that when there exists an $\ell$-repeat-free sequence of length $n$ over $\Sigma_q$, an $S$-uniquely-decodable $\ell$-labeling code $\cC$ is of maximum size if and only if $S$ is $\ell$-repeat-free. However, for any sequence $S$ of length $n$, if $S$ is $\ell$-repeat-free then by~\autoref{def: de~Bruijn} and~\autoref{remark: de~Bruijn}, ${|S|\le q^\ell+\ell-1}$. That is, \autoref{const: meet uppper bound} is relevant only for $\ell\ge \log_q(|S|-\ell+1)$. Moreover, as discussed in~\autoref{def: de~Bruijn} and~\autoref{remark: de~Bruijn}, an $\ell$-repeat-free sequence of length $n$ over $\Sigma_q$ exists if and only if $\ell\ge \log_q(n-\ell+1)$  which  leads to the following consequences.

\begin{corollary}\label{cor: optimal code for S}
    Assume $n,\ell,$ and $q>1$ be integers. 
    \begin{enumerate}
        \item If $\ell< \log_q(n-\ell+1)$ then $M_\ell(n,q)<|\cR_\ell(n)|.$ 
        \item If $\ell\geq \log_q(n-\ell+1)$ then $M_\ell(S) = M_\ell(n,q)$ if and only if $S$ is an $\ell$-repeat-free sequence.
        \item $M_\ell(n,q) = |\cR_\ell(n)|$ if and only if $\ell\geq \log_q(n-\ell+1)$.
    \end{enumerate}
     \hfill\quad $\blacksquare$
\end{corollary}

\begin{corollary}
      Assume $n,\ell,$ and $q>1$ be integers such that $\ell\ge \log_q(n-\ell+1)$. It holds that
      \[
    \limsup_{n\to\infty}
    \frac{\log_2 M_\ell(n,q)}{n}
    =
    \log_2\rho_\ell,
\]
where $\rho_\ell\in(1,2]$ is the unique real solution in this interval of
\[
    \rho_\ell^{\ell-1}(\rho_\ell-1)^2=1.
\]  \hfill\quad $\blacksquare$
\end{corollary}

The distribution of the shortest run of ones in binary sequences was studied in~\cite{makri2007shortest}. The binomial-sum formula in~\autoref{lem: eta} was established in the ISIT 2024 conference version of this work. Subsequently, Balado and Silvestre~\cite[Eq.~(61)]{balado2026systematic} gave an expression for the number of length-$n$ binary sequences containing no runs of ones of lengths $1,\ldots,k$. For $\ell\ge2$, setting $k=\ell-1$ in their expression gives a different representation of $|\cR_\ell(n)|$. For completeness, we include our combinatorial proof in the appendix.

\begin{lemma}\label{lem: eta} For any two positive integers $n,\ell$, it holds that 

\[
|\cR_\ell(n)|= \sum_{t=0}^{\lfloor \frac{n+1}{\ell+1}\rfloor}\binom{n-t(\ell-1)+1}{2t}.
\]  \hfill\quad $\blacksquare$
\end{lemma}

The asymptotic bound in \autoref{cor: RLL capacity} is obtained for a fixed value of $\ell$. However, for a fixed alphabet size $q$, the existence of an $\ell$-repeat-free sequence of length $n$ requires
\[
\ell\ge\log_q(n-\ell+1).
\]
Thus, when $n$ grows, the label length used in \autoref{const: meet uppper bound} must also grow. Consequently, the fixed-$\ell$ capacity does not describe the asymptotic size of the codes obtained by this construction. We therefore analyze $|\cR_\ell(n)|$ in the logarithmic regime
\[
\ell=\left\lceil c\log_q n\right\rceil,
\qquad c\ge1,
\]
using the expression given in \autoref{lem: eta}.

\begin{theorem}\label{th: logarithmic label length}
Let $q>1$ and $c\ge1$ be constants. For every positive integer $n$, set
\[
\ell=\left\lceil c\log_q n\right\rceil.
\]
Let $S\in\Sigma_q^n$ be an $\ell$-repeat-free sequence, and let $\cC_{(n,\ell)}$ be the code obtained by \autoref{const: meet uppper bound}. Then
\[
M_\ell(n,q)
=
M_\ell(S)
=
|\cC_{(n,\ell)}|
=
|\cR_\ell(n)|
=
2^{\Theta\left(
\frac{\log\log n}{\log n}\,n
\right)}
\]
as $n\to\infty$, where the constants implicit in the $\Theta$ notation may depend on $q$ and $c$.
\end{theorem}

\begin{IEEEproof}
Since $c\ge1$, for all sufficiently large $n$,
\[
\ell
\ge
\log_q n
\ge
\log_q(n-\ell+1).
\]
Therefore, \autoref{cor: optimal code for S} gives
\[
M_\ell(n,q)=M_\ell(S)=|\cR_\ell(n)|,
\]
while \autoref{th: code for l-repeat-free} gives
\[
|\cC_{(n,\ell)}|=|\cR_\ell(n)|.
\]
It remains to estimate $|\cR_\ell(n)|$.

By \autoref{lem: eta},
\[
|\cR_\ell(n)|
=
\sum_{t=0}^{\left\lfloor\frac{n+1}{\ell+1}\right\rfloor}
\binom{n-t(\ell-1)+1}{2t}.
\]
Set $T\triangleq
\left\lfloor\frac{n+1}{\ell+1}\right\rfloor.$
Since $T=\Theta(n/\ell)$ and $2T\le n+1$, we have
\[
|\cR_\ell(n)| \le
\sum_{t=0}^{T}\binom{n+1}{2t}\le
\sum_{k=0}^{2T}\binom{n+1}{k}.
\]
Moreover, as $2T/(n+1)\le 1$, we have
\begin{align*}
\left(\frac{2T}{n+1}\right)^{2T}
\sum_{k=0}^{2T}\binom{n+1}{k}
&\le
\sum_{k=0}^{2T}
\binom{n+1}{k}
\left(\frac{2T}{n+1}\right)^k\\
&\le
\left(1+\frac{2T}{n+1}\right)^{n+1}\\
&\le e^{2T}.
\end{align*}
Therefore,
\[
|\cR_\ell(n)|
\le
\left(\frac{e(n+1)}{2T}\right)^{2T}
=
2^{O\left((n/\ell)\log\ell\right)}.
\]

For the reverse inequality, take
\[
t=\left\lfloor\frac{n}{2\ell}\right\rfloor.
\]
For all sufficiently large $n$, we have $1\le t\le T$. Using the term corresponding to this value of $t$ in \autoref{lem: eta} and the
bound $\binom NK\ge(N/K)^K$, we obtain
\[
|\cR_\ell(n)|
\ge
\binom{n-t(\ell-1)+1}{2t}
\ge
\left(
\frac{n-t(\ell-1)+1}{2t}
\right)^{2t}.
\]
Since $t\le n/(2\ell)$,
\[
n-t(\ell-1)+1\ge\frac n2
\]
and
\[
2t\le\frac n\ell.
\]
Therefore,
\[
|\cR_\ell(n)|
\ge
\left(\frac{\ell}{2}\right)^{2t}
=
2^{\Omega\left((n/\ell)\log\ell\right)},
\]
where the last equality follows from
$t=\Theta(n/\ell)$ and
$\log(\ell/2)=\Theta(\log\ell)$.

Consequently,
\[
|\cR_\ell(n)|
=
2^{\Theta\left((n/\ell)\log\ell\right)}.
\]
Finally, $\ell=\Theta(\log n)$ and $\log\ell=\Theta(\log\log n)$, which proves the claim.
\end{IEEEproof}

\subsection{Non-overlapping and partially non-overlapping codes}\label{subsec: overlap models}

We first derive an upper bound for the non-overlapping model and show that it is attained by every $\ell$-repeat-free reference sequence. For this purpose, let $\cR_\ell^{\mathsf{NO}}(n)$ denote the set of binary sequences of length $n$ in which the length of every run of
\emph{ones} is divisible by $\ell$.

\begin{theorem}\label{th: repeat-free overlap models}
For every $S\in\Sigma_q^n$,
\[
    M_\ell^{\mathsf{NO}}(S)
    \le
    \left|\cR_\ell^{\mathsf{NO}}(n)\right|.
\]
Additionally, if $S$ is $\ell$-repeat-free, then
\[
    M_\ell^{\mathsf{NO}}(S)
    =
    M_\ell^{\mathsf{PNO}}(S)
    =
    \left|\cR_\ell^{\mathsf{NO}}(n)\right|.
\]
\end{theorem}

\begin{IEEEproof}
Let $\cC$ be an $S$-uniquely-decodable, $S$-non-overlapping $\ell$-labeling code, and let $\Lambda\in\cC$. For every occurrence $o=(\lambda,i)\in\cO(\Lambda,S)$, we have $|\lambda|=\ell$, and hence
\[
    |\mathsf{supp}(o)|=\ell.
\]
Moreover, since $\Lambda$ is $S$-non-overlapping, the supports $\{\mathsf{supp}(o):o\in\cO(\Lambda,S)\}$ are pairwise disjoint. Therefore, every run of \emph{ones} in $\Lambda(S)$ is formed by adjacent disjoint supports of length $\ell$, and its length is divisible by $\ell$. Thus,
\[
    \Lambda(S)\in\cR_\ell^{\mathsf{NO}}(n).
\]
Since $\cC$ is $S$-uniquely-decodable, its codesets induce distinct
labelings, and consequently
\[
    |\cC|\le\left|\cR_\ell^{\mathsf{NO}}(n)\right|.
\]

Assume now that $S$ is $\ell$-repeat-free. For every $X\in\cR_\ell^{\mathsf{NO}}(n)$, partition each run of \emph{ones} in $X$ into consecutive blocks of length $\ell$, and let $\Lambda_X$ consist of the corresponding length-$\ell$ substrings of $S$. Since $S$ is $\ell$-repeat-free, every label in $\Lambda_X$ occurs exactly once in $S$. Therefore, $\Lambda_X$ is $S$-non-overlapping and $\Lambda_X(S)=X$. This proves that
\[
M_\ell^{\mathsf{NO}}(S)
=
\left|\cR_\ell^{\mathsf{NO}}(n)\right|.
\]

Finally, when $S$ is $\ell$-repeat-free, every nonredundant length-$\ell$ label has exactly one occurrence in $S$. Hence, an $S$-partially-non-overlapping codeset is necessarily $S$-non-overlapping, which proves the second equality.
\end{IEEEproof}

\begin{corollary}\label{cor: strict overlap models}
If $S$ is $\ell$-repeat-free, $\ell>1$, and $n>\ell$, then
\[
M_\ell^{\mathsf{NO}}(S)
=
M_\ell^{\mathsf{PNO}}(S)
<
M_\ell(S).
\]
\end{corollary}

\begin{IEEEproof}
By \autoref{th: repeat-free overlap models}, the first two quantities are equal to $|\cR_\ell^{\mathsf{NO}}(n)|$, while \autoref{th: code for l-repeat-free} gives $M_\ell(S)=|\cR_\ell(n)|$. Since
\[
1^{\ell+1}0^{n-\ell-1}
\in
\cR_\ell(n)\setminus\cR_\ell^{\mathsf{NO}}(n),
\]
the inequality is strict.
\end{IEEEproof}

The preceding result settles both restricted models when the reference sequence is $\ell$-repeat-free. For a general reference sequence, a label may have several occurrences, and the two models need not coincide. We next give a sufficient condition for an $S$-non-overlapping code to be $S$-uniquely-decodable.

\begin{lemma}
    Let $\cC$ be an $\ell$-labeling code. For any $S\in\Sigma_q^*$, if $\cC$ is an $S$-non-overlapping  code, then $\cC$ is $S$-uniquely-decodable.
\end{lemma}
\begin{IEEEproof}
    Let $S\in\Sigma_q^n$ for some integers $q,n>1$ and let $\cC$ be an $\ell$-labeling $S$-non-overlapping code.   For any $\Lambda\in\cC$, $\Lambda$ is an $S$-non-overlapping codeset. Hence, the sequence $\Lambda(S)$ is a binary sequence of length $n$ in which any run of \emph{ones} is of length~$j\ell$ for some integer $j\ge 1$. Furthermore, for any run of \emph{ones} in $\Lambda(S)$ there is a unique way to partition it to non-overlapping segments of length $\ell$. Thus, for any run of $j\ell$ \emph{ones} in $\Lambda(S)$, $j\ge 1$, we have a unique set of non-overlapping labels of length $\ell$ that can create this run. That is, $\Lambda$ can be uniquely determined given $S$ and $\Lambda(S)$, which completes the proof. 
\end{IEEEproof}

We now turn to the construction of labeling codes for reference sequences that are not necessarily $\ell$-repeat-free. This setting is unavoidable when
\[
    \ell<\log_q(n-\ell+1),
\]
since no $\ell$-repeat-free sequence of length $n$ exists over $\Sigma_q$. Our construction uses a non-overlapping code $\cN_\ell\subseteq\Sigma_q^\ell$ in the sense of \cite{levy2018mutually,levenshtein1970maximum, yazdi2018mutually,blackburn2015non}; that is, no nonempty proper prefix of a label in $\cN_\ell$ is a suffix of a label in $\cN_\ell$, including the same label. Consequently, distinct occurrences of labels from $\cN_\ell$ cannot overlap in any reference sequence.

\begin{construction} \label{cons: not rf}
Given a sequence $S\in\Sigma_q^n$ and a positive integer $\ell\le n$, construct an $\ell$-labeling code $\cC$ as follows.
\begin{enumerate}
    \item 
    Let $\cN_\ell\subseteq\Sigma_q^\ell$ be a non-overlapping code.
    \item Let $\cN_{(\ell,S)}\triangleq \cN_\ell\cap W_\ell(S)$.
    \item  Let $\cC\triangleq \cP(\cN_{(\ell,S)})$, where $ \cP(\cN_{(\ell,S)})$ is the power set of $\cN_{(\ell,S)}$.
\end{enumerate}
    \hfill\quad $\blacksquare$
\end{construction}

It was proved in~\cite{blackburn2015non} that there exists a non-overlapping code $\cN_\ell$ such that $|\cN_\ell|\ge c_1q^\ell/\ell$ for some absolute constant $c_1$. This implies the following lower bounds for all three labeling models, which can be readily verified.

\begin{claim}\label{cla: non overlapping}
For any $S\in\Sigma_q^n$,
\[
M_\ell(S)
\ge
M_\ell^{\mathsf{PNO}}(S)
\ge
M_\ell^{\mathsf{NO}}(S)
\ge
2^{|\cN_{(\ell,S)}|}.
\]
In particular, if $n\ge q^\ell+\ell-1$, then
\[
M_\ell(n,q)
\ge
M_\ell^{\mathsf{PNO}}(n,q)
\ge
M_\ell^{\mathsf{NO}}(n,q)
\ge
2^{c_1q^\ell/\ell}
\]
for some absolute constant $c_1$.
\end{claim}

To have a specific value (without the unknown constant $c_1$), one can use Levenshtein's construction~\cite{levenshtein1970maximum}, which is a special case of the construction of~\cite{blackburn2015non}. This will result in a lower bound of $\frac{(q-1)^2(2q-1)}{4\ell q^4}q^\ell$.

While \autoref{cons: not rf} gives lower bounds for all three models, the bound on the unrestricted quantity $M_\ell(n,q)$ can be improved when the reference sequence may be selected, as shown in the next theorem.

\begin{theorem}\label{th: de Bruijn prefix lower bound}
Let $n,\ell$ be positive integers and let $q>1$. If
$\ell<\log_q(n-\ell+1)$,
then
\[
M_\ell(n,q)
\ge
|\cR_\ell(q^\ell+\ell-1)|
=
2^{\Theta\left(
\frac{\log\ell}{\ell}\,q^\ell
\right)},
\]
where the asymptotic equality holds for fixed $q$ as
$\ell\to\infty$.
\end{theorem}

\begin{IEEEproof}
Let
\[
    N=q^\ell+\ell-1.
\]
Choose a symbol $\sigma\in\Sigma_q$ and an unwrapped de~Bruijn sequence $S'\in\Sigma_q^N$ of order $\ell$ whose final $\ell$ symbols are $\sigma^\ell$. Such a sequence can be obtained (see~\autoref{def: de~Bruijn}).

Since $\ell<\log_q(n-\ell+1)$, we have $N<n$. Define
\[
    S=S'\sigma^{n-N}.
\]
Apply \autoref{const: meet uppper bound} to $S'$ and let $\cC_{(N,\ell)}$ be the resulting code. For every
$X\in\cR_\ell(N)$, we have
\[
    \Lambda_X(S')=X.
\]

Every length-$\ell$ substring of $S$ that crosses the end of $S'$ or is contained in its appended suffix is equal to $\sigma^\ell$. Moreover, the occurrence of $\sigma^\ell$ in $S'$ occupies its final $\ell$ positions. Therefore, any additional occurrences of this label do not change the first $N$ entries of the labeling, and
\[
    \Lambda_X(S)_{[1;N]}=X.
\]
Consequently, distinct sequences in $\cR_\ell(N)$ induce distinct labelings of $S$. Hence,
\[
    M_\ell(n,q)
    \ge
    |\cC_{(N,\ell)}|
    =
    |\cR_\ell(N)|.
\]
Since $N=\Theta(q^\ell)$ for fixed $q$, we have $\ell=\Theta(\log N)$. The estimate established in the proof of
\autoref{th: logarithmic label length} therefore gives
\[
|\cR_\ell(N)|
=
2^{\Theta\left((N/\ell)\log\ell\right)}
=
2^{\Theta\left((q^\ell/\ell)\log\ell\right)}.
\]
\end{IEEEproof}

\section{Conclusions}\label{sec:conclusion}
In this work, we introduced a framework for representing information by labeling a known DNA reference sequence. Rather than synthesizing a new molecule for every message, the information is represented by selecting a codeset whose labels induce a distinct binary pattern on the reference sequence. We formulated the corresponding maximum-code problems for a given reference sequence and for an optimized reference sequence, and extended them to restricted collections of executable labels. We also introduced the non-overlapping and partially-non-overlapping models, which account for physical restrictions on simultaneously labeling overlapping sites.

We obtained exact results in several settings. For periodic reference sequences, we determined the maximum code size for periods one and two and derived a tight bound when the least period is no larger than the alphabet size, with equality when the symbols within one period are distinct. For fixed-length labels, we showed that every possible labeling belongs to the constrained system $\cR_\ell$, analyzed the cardinality and capacity of this system, and characterized precisely when the resulting upper bound is attained. In particular, an $\ell$-repeat-free reference sequence is optimal, and our construction realizes every sequence in $\cR_\ell(n)$. When $\ell=\Theta(\log n)$, this gives codes of size
\[
2^{\Theta\left(
        \frac{\log\log n}{\log n}\,n
    \right)}.
\]
We also considered the regime in which an $\ell$-repeat-free reference sequence does not exist and provided lower bounds based on non-overlapping codes and on reference sequences containing a de~Bruijn prefix. Finally, for $\ell$-repeat-free reference sequences, we determined the maximum sizes under the non-overlapping and partially-non-overlapping restrictions and showed that these two quantities coincide, while they are strictly smaller than the unrestricted maximum when $\ell>1$ and $n>\ell$.

Several fundamental questions remain open. First, it would be desirable to obtain structural conditions characterizing when a labeling code is $S$-uniquely-decodable for a general reference sequence $S$. The non-overlapping condition gives a useful sufficient condition, but it is not necessary. A related inverse problem is to determine whether, for a given labeling code $\cC$, there exists a reference sequence for which $\cC$ is uniquely decodable, and, if so, to find the shortest such sequence. For a restricted set of executable labels $\cV$, it remains to characterize the subsets of $\cP(\cV)$ that form good labeling codes, determine the corresponding maximum code sizes, and design reference sequences attaining them. Such results would make it possible to incorporate constraints imposed by particular labeling techniques, including length, sequence-composition, and PAM requirements.

Another central open problem concerns fixed-length labeling when the reference sequence is not $\ell$-repeat-free. In this regime, the general upper bound $|\cR_\ell(n)|$ cannot be attained, while the constructions presented here provide only lower bounds. Closing this gap, characterizing the optimal reference sequences, and determining whether constructions that allow overlaps can improve upon the non-overlapping construction are natural next steps. It would also be useful to obtain explicit finite-length bounds and constants, particularly for the DNA alphabet. Similarly, the periodic analysis can be extended to periods larger than the alphabet size, where symbols must repeat within a period.

The non-overlapping and partially-non-overlapping models also leave several extremal questions. For general reference sequences, tight bounds on $M_\ell^{\mathsf{NO}}(S)$ and $M_\ell^{\mathsf{PNO}}(S)$, and on their maxima over the reference sequence, are not yet known. It remains to characterize when the inequalities
\[
M_\ell(S)
\ge
M_\ell^{\mathsf{PNO}}(S)
\ge
M_\ell^{\mathsf{NO}}(S)
\]
are equalities, and when either inequality is strict. The partially-non-overlapping model is especially interesting when labels have repeated occurrences, since an isolated occurrence may identify a label even when its other occurrences overlap. Studying these models for periodic reference sequences and developing a more detailed model for competition between labels at overlapping sites are further
directions of interest.

Finally, the present work considers an ideal noiseless labeling process. A practical system must account for absent labels, off-target labeling at similar sequences, and errors in locating the labeled positions. These phenomena motivate sequence-dependent error models and labeling codes capable of correcting the resulting errors. They also raise a new reference-design question: although repeat-free sequences are optimal in the noiseless fixed-length setting, repeated occurrences of a label may provide useful redundancy when some occurrences are missing. On the other hand, repeated or similar substrings may increase the possibility of off-target labeling.
Understanding this tradeoff and determining whether repeat-free reference sequences remain optimal in noisy settings constitute promising directions for future work. The research on error correction is currently under consideration.

\balance
\bibliographystyle{IEEEtran}
\bibliography{refs}

\begin{thebibliography}{10}
\providecommand{\url}[1]{#1}
\csname url@samestyle\endcsname
\providecommand{\newblock}{\relax}
\providecommand{\bibinfo}[2]{#2}
\providecommand{\BIBentrySTDinterwordspacing}{\spaceskip=0pt\relax}
\providecommand{\BIBentryALTinterwordstretchfactor}{4}
\providecommand{\BIBentryALTinterwordspacing}{\spaceskip=\fontdimen2\font plus
\BIBentryALTinterwordstretchfactor\fontdimen3\font minus
  \fontdimen4\font\relax}
\providecommand{\BIBforeignlanguage}[2]{{%
\expandafter\ifx\csname l@#1\endcsname\relax
\typeout{** WARNING: IEEEtran.bst: No hyphenation pattern has been}%
\typeout{** loaded for the language `#1'. Using the pattern for}%
\typeout{** the default language instead.}%
\else
\language=\csname l@#1\endcsname
\fi
#2}}
\providecommand{\BIBdecl}{\relax}
\BIBdecl

\bibitem{barlev24isit}
D.~Bar-Lev, T.~Etzion, E.~Yaakobi, and Z.~Yakhini, ``Representing information
  on {DNA} using patterns induced by enzymatic labeling,'' in \emph{2024 IEEE
  International Symposium on Information Theory (ISIT)}, 2024, pp. 1943--1948.

\bibitem{moter2000fluorescence}
A.~Moter and U.~B. G{\"o}bel, ``Fluorescence in situ hybridization ({FISH}) for
  direct visualization of microorganisms,'' \emph{Journal of microbiological
  methods}, vol.~41, no.~2, pp. 85--112, 2000.

\bibitem{chen2018efficient}
B.~Chen, W.~Zou, H.~Xu, Y.~Liang, and B.~Huang, ``Efficient labeling and
  imaging of protein-coding genes in living cells using {CRISPR}-tag,''
  \emph{Nature communications}, vol.~9, no.~1, p. 5065, 2018.

\bibitem{gruszka2021single}
D.~Gruszka, J.~Jeffet, S.~Margalit, Y.~Michaeli, and Y.~Ebenstein,
  ``Single-molecule optical genome mapping in nanochannels:
  {Multidisciplinarity} at the nanoscale,'' \emph{Essays in Biochemistry},
  vol.~65, no.~1, pp. 51--66, 2021.

\bibitem{ma2015multicolor}
H.~Ma, A.~Naseri, P.~Reyes-Gutierrez, S.~A. Wolfe, S.~Zhang, and T.~Pederson,
  ``Multicolor {CRISPR} labeling of chromosomal loci in human cells,''
  \emph{Proceedings of the National Academy of Sciences}, vol. 112, no.~10, pp.
  3002--3007, 2015.

\bibitem{deen2017methyltransferase}
J.~Deen, C.~Vranken, V.~Leen, R.~K. Neely, K.~P. Janssen, and J.~Hofkens,
  ``Methyltransferase-directed labeling of biomolecules and its applications,''
  \emph{Angewandte Chemie International Edition}, vol.~56, no.~19, pp.
  5182--5200, 2017.

\bibitem{young2020technical}
A.~P. Young, D.~J. Jackson, and R.~C. Wyeth, ``A technical review and guide to
  {RNA} fluorescence in situ hybridization,'' \emph{PeerJ}, vol.~8, p. e8806,
  2020.

\bibitem{assa2024quantifying}
G.~Assa, N.~Kalter, M.~Rosenberg, A.~Beck, O.~Markovich, T.~Gontmakher,
  A.~Hendel, and Z.~Yakhini, ``Quantifying allele-specific {CRISPR} editing
  activity with {CRISPECTOR}2. 0,'' \emph{Nucleic Acids Research}, vol.~52,
  no.~16, pp. e78--e78, 2024.

\bibitem{tabatabaei2020dna}
S.~K. Tabatabaei, B.~Wang, N.~B.~M. Athreya, B.~Enghiad, A.~G. Hernandez, C.~J.
  Fields, J.-P. Leburton, D.~Soloveichik, H.~Zhao, and O.~Milenkovic, ``{DNA}
  punch cards for storing data on native {DNA} sequences via enzymatic
  nicking,'' \emph{Nature communications}, vol.~11, no.~1, p. 1742, 2020.

\bibitem{Sadremomtaz2023storageCRISPR}
A.~Sadremomtaz, R.~F. Glass, J.~E. Guerrero, D.~R. LaJeunesse, E.~A. Josephs,
  and R.~Zadegan, ``{Digital data storage on {DNA} tape using CRISPR base
  editors},'' \emph{Nature Communications}, vol.~14, no.~1, p. 6472, 2023.

\bibitem{jinek2012programmable}
M.~Jinek, K.~Chylinski, I.~Fonfara, M.~Hauer, J.~A. Doudna, and E.~Charpentier,
  ``A programmable dual-{RNA}--guided {DNA} endonuclease in adaptive bacterial
  immunity,'' \emph{science}, vol. 337, no. 6096, pp. 816--821, 2012.

\bibitem{cong2013multiplex}
L.~Cong, F.~A. Ran, D.~Cox, S.~Lin, R.~Barretto, N.~Habib, P.~D. Hsu, X.~Wu,
  W.~Jiang, L.~A. Marraffini \emph{et~al.}, ``Multiplex genome engineering
  using {CRISPR/Cas} systems,'' \emph{Science}, vol. 339, no. 6121, pp. 819--823,
  2013.

\bibitem{amit2021crispector}
I.~Amit, O.~Iancu, A.~Levy-Jurgenson, G.~Kurgan, M.~S. McNeill, G.~R. Rettig,
  D.~Allen, D.~Breier, N.~Ben~Haim, Y.~Wang \emph{et~al.}, ``{CRISPECTOR}
  provides accurate estimation of genome editing translocation and off-target
  activity from comparative {NGS} data,'' \emph{Nature Communications}, vol.~12,
  no.~1, p. 3042, 2021.

\bibitem{hanania2023capacity}
D.~Hanania, D.~Bar-Lev, Y.~Nogin, Y.~Shechtman, and E.~Yaakobi, ``On the
  capacity of {DNA} labeling,''  \emph{IEEE Transactions on Information Theory}, vol.~71, no.~5, pp. 3457--3472, 2025.
  
  \bibitem{hofmeister2026achieving}
C.~Hofmeister, A.~Gruica, D.~Hanania, R.~Bitar, and E.~Yaakobi,
``Achieving DNA labeling capacity with minimum labels through extremal
de Bruijn subgraphs,''
\emph{IEEE Trans. Inf. Theory}, vol.~72, no.~2, pp.~1122--1132,
Feb.~2026, doi: 10.1109/TIT.2025.3640866.

\bibitem{wu2026capacity}
Z.~Wu, Q.~Cao, L.~Liu, and B.~Bai,
``On the capacity of DNA labeling in the single-label setting,''
\emph{arXiv preprint arXiv:2609.36715}, 2026.

\bibitem{hanania2026composite}
D.~Hanania, T.~T.~Nguyen, K.~Cai, E.~Yaakobi, and Y.~M.~Chee,
``DNA labeling with composite symbols,''
in \emph{Proc. 2026 IEEE Int. Symp. Inf. Theory (ISIT)},
Guangzhou, China, 2026, pp.~1--6,
doi: 10.1109/ISIT62367.2026.11653892.

\bibitem{hanania2025error}
D.~Hanania and E.~Yaakobi,
``Error-correcting codes for labeled DNA sequences,''
in \emph{Proc. 2025 IEEE Int. Symp. Inf. Theory (ISIT)},
Ann Arbor, MI, USA, Jun.~2025, pp.~1--6,
doi: 10.1109/ISIT63088.2025.11195710.

\bibitem{liu2026burst}
W.~Liu, Z.~Jiang, Z.~Huang, and H.~Hou,
``Asymptotically optimal codes for correcting burst deletions and insertions
in labeled DNA sequences,''
\emph{arXiv preprint arXiv:2606.14573}, 2026.


\bibitem{Nogin2023OM}
Y.~Nogin, D.~Bar-Lev, D.~Hanania, T.~Detinis~Zur, Y.~Ebenstein, E.~Yaakobi,
  N.~Weinberger, and Y.~Shechtman, ``{Design of optimal labeling patterns for
  optical genome mapping via information theory},'' \emph{Bioinformatics},
  vol.~39, no.~10, p. btad601, 09 2023.

\bibitem{levy2013beyond}
M.~Levy-Sakin and Y.~Ebenstein, ``Beyond sequencing: {Optical} mapping of {DNA}
  in the age of nanotechnology and nanoscopy,'' \emph{Current opinion in
  biotechnology}, vol.~24, no.~4, pp. 690--698, 2013.

\bibitem{muller2017optical}
V.~M{\"u}ller and F.~Westerlund, ``Optical {DNA} mapping in nanofluidic
  devices: {Principles} and applications,'' \emph{Lab on a Chip}, vol.~17,
  no.~4, pp. 579--590, 2017.

\bibitem{zhang2024parallel}
C.~Zhang, R.~Wu, F.~Sun, Y.~Lin, Y.~Liang, J.~Teng, N.~Liu, Q.~Ouyang, L.~Qian, and H.~Yan,
``Parallel molecular data storage by printing epigenetic bits on {DNA},''
\emph{Nature}, vol.~634, pp.~824--832, 2024.

\bibitem{marcus2001introduction}
\BIBentryALTinterwordspacing
B.~H. Marcus, R.~M. Roth, and P.~H. Siegel, ``An introduction to coding for
  constrained systems,'' San Diego, CA, USA:Lecture Notes, 2001. [Online].
  Available: \url{http://www.math.ubc.ca/~marcus/Handbook/}
\BIBentrySTDinterwordspacing

\bibitem{elishco2021repeat}
O.~Elishco, R.~Gabrys, E.~Yaakobi, and M.~M{\'e}dard, ``Repeat-free codes,''
  \emph{IEEE Transactions on Information Theory}, vol.~67, no.~9, pp.
  5749--5764, 2021.

\bibitem{van1987circuits}
T.~van Aardenne-Ehrenfest and N.~G. de~Bruijn, ``Circuits and trees in oriented
  linear graphs,'' in \emph{Classic papers in combinatorics}.\hskip 1em plus
  0.5em minus 0.4em\relax Springer, 1987, pp. 149--163.

\bibitem{etzion2024sequences}
T.~Etzion, \emph{Sequences and the de Bruijn Graph}.\hskip 1em plus 0.5em minus
  0.4em\relax Elsevier, 2024.

\bibitem{adler1983algorithms}
R.~Adler, D.~Coppersmith, and M.~Hassner, ``Algorithms for sliding block
  codes-an application of symbolic dynamics to information theory,'' \emph{IEEE
  Transactions on Information Theory}, vol.~29, no.~1, pp. 5--22, 1983.

\bibitem{makri2007shortest}
F.~S.~Makri, A.~N.~Philippou, and Z.~M.~Psillakis,
``Shortest and longest length of success runs in binary sequences,''
\emph{Journal of Statistical Planning and Inference}, vol.~137, no.~7, pp.~2226--2239, 2007,
doi: 10.1016/j.jspi.2006.07.007.

\bibitem{balado2026systematic}
F.~Balado and G.~C.~M.~Silvestre,
``Systematic enumeration of fundamental quantities involving runs in binary strings,''
arXiv:2602.10005, 2026,
doi: 10.48550/arXiv.2602.10005.

\bibitem{levy2018mutually}
M.~Levy and E.~Yaakobi, ``Mutually uncorrelated codes for {DNA} storage,''
  \emph{IEEE Transactions on Information Theory}, vol.~65, no.~6, pp.
  3671--3691, 2019.

\bibitem{levenshtein1970maximum}
V.~I. Levenshtein, ``Maximum number of words in codes without overlaps,''
  \emph{Problemy Peredachi Informatsii}, vol.~6, no.~4, pp. 88--90, 1970.

\bibitem{yazdi2018mutually}
S.~T. Yazdi, H.~M. Kiah, R.~Gabrys, and O.~Milenkovic, ``Mutually uncorrelated
  primers for {DNA}-based data storage,'' \emph{IEEE Transactions on
  Information Theory}, vol.~64, no.~9, pp. 6283--6296, 2018.

\bibitem{blackburn2015non}
S.~R. Blackburn, ``Non-overlapping codes,'' \emph{IEEE Transactions on
  Information Theory}, vol.~61, no.~9, pp. 4890--4894, 2015.

\end{thebibliography}
\clearpage

\appendix

\noindent \textit{Proof of~\autoref{lem: period bound}}:
If the same symbol occurs in two positions of the period, replace one of these occurrences, in every period, by a new symbol. For each codeset, and for any label that contains $t$ copies of the repeated symbol $\sigma$, we replace this label with all $2^t$ labels containing all combinations of replacing $\sigma$ with the new symbol. Hence, splitting a repeated symbol cannot decrease $M(S)$. Repeating this argument, and using $\pi\le q$, it is sufficient to assume that $S_{[1;\pi]}=\sigma_1\sigma_2\cdots\sigma_\pi$, where $\sigma_1,\ldots,\sigma_\pi$ are distinct. Note that this process can introduce redundant labels; however, this cannot reduce the size of the code.

To find $M(S)$, we consider several distinct cases. If $n=\pi$ then all the symbols of $S=\sigma_1\ldots\sigma_\pi$ are distinct and any generated codeset ${\Lambda \in \cP(\{\sigma_1,\ldots,\sigma_\pi\})}$ yields a different labeling of $S$, i.e.,  $M(S)\ge |\cP(\{\sigma_1,\ldots,\sigma_\pi\})| = 2^\pi$. Additionally, since $\Lambda(S)$ is a binary sequence of length $|S|=\pi$, we have that $M(S)\le 2^\pi$. Hence, for $n=\pi$ we have that $M(S)=2^\pi$.

        If $n=2\pi$ we  consider two types of codesets. 
        \begin{itemize}
            \item If $\Lambda$ is a codeset such that $\Lambda(S)_{[1;\pi]} = \Lambda(S)_{[\pi+1;\pi]}$ then similarly to the case where $n=\pi$, $\Lambda(S)$ is determined by its first $\pi$ bits which implies that there are at most $2^\pi$ such codesets. Moreover, for any distinct ${\Lambda_1,\Lambda_2 \in \cP(\{\sigma_1,\ldots,\sigma_\pi\})}$ we have that $\Lambda_1\not\equiv_S\Lambda_2$ and $\Lambda_i(S)_{[1;\pi]} = \Lambda_i(S)_{[\pi+1;\pi]}$, ${i=1,2}$.  That is, we have exactly $2^\pi$ such  codesets that can be used.
            \item Otherwise, it can be verified that  $\Lambda(S)_{[1;\pi]} \ne \Lambda(S)_{[\pi+1;\pi]}$ cannot hold if all the labels $\lambda\in\Lambda$ are substring of $\Lambda(S)_{[1;\pi]}$. That is, we may only consider codesets $\Lambda$ such that there exists $\lambda\in\Lambda$ which is a substring of $S$ but not a substring of $\Lambda(S)_{[1;\pi]}$.
            Such a label $\lambda$ must start at some symbol $\sigma_i$ that belongs to the first $\pi$ symbols of $S$ and end in a symbol $\sigma_j$ that belongs to the last $\pi$ symbols of~$S$. Next, we consider all different cases depending on $i$ and $j$.
            \begin{enumerate}
                \item[1.] If $i=1$ and $j=\pi$ then $\lambda(S)=\Lambda(S)$ is the all-ones sequence which does not satisfy $\Lambda(S)_{[1;\pi]} \ne \Lambda(S)_{[\pi+1;\pi]}$.
                \item[2.] If $i=1$ and $j<\pi$ then $\lambda(S)=1^{\pi+j}0^{\pi-j}$. 
                We first note that if $\Lambda$ contains additional labels with the symbol $\sigma_{j+1}$ (that are substring of $S$), then $\Lambda(S)$ will be a sequence that starts with a prefix of more than $\pi+j$ ones -- a case that is  accounted for when considering a larger value of $j$. Hence we assume $\Lambda$ contains $\lambda$ and any other label in $\Lambda$ may only be a substring of the last $\pi-j-1$ symbols of $S$. 
                This contributes $2^{\pi-j-1}$ additional codesets for each $1\le j< \pi$ and 
                $$\sum_{j=1}^{\pi-1}2^{\pi-j-1} = 2^{\pi-1}\cdot \frac{2^\pi-2}{2^\pi} = 2^{\pi-1}-1$$
                additional codesets in total. 
                \item[3.] If $i>1$ and $j=\pi$ then $\lambda(S) = 0^{i-1}1^{2\pi-i+1}$. Similarly to the previous case, we will restrict ourselves to codesets in which the length of the prefix of ones does not change, which results in $2^{i-2}$ codesets for each $1<i\le \pi$ and  $\sum_{i=2}^\pi 2^{i-2} = 2^{\pi-1}-1$ additional codesets in total. 
                \item[4.] If $i>1$ and $j<\pi$ then ${\lambda(S)= 0^{i-1}1^{\pi-i+1+j}0^{\pi-j}}$ and similarly to the previous two cases, we will keep the length of the middle substring of ones fixed, i.e., we avoid using the symbols $\sigma_{i-1}$ and $\sigma_{j+1}$ in our additional labels. To this end, we need to consider the following disjoint cases.
                \begin{enumerate}
                    \item[(a)] If $i<  j+2$ then the first $i-2$ symbols and the last $\pi-j-1$ symbols are different, i.e., \begin{align}\label{eq: symbol intersection}
                        \{\sigma_1,\ldots,\sigma_{i-2}\}\cap\{\sigma_{j+2},\ldots, \sigma_\pi\} = \varnothing.
                    \end{align}
                    Moreover, $\sigma_{i-1}$ and $\sigma_{j+1}$ are not part of these two sets of symbols. 
                    Hence, similarly to Case~2 all possible $2^{\pi-j-1}$ binary suffixes of length $\pi-j-1$ are possible and since the symbols at the beginning of $S$ are disjoint to the ones we used, for each such suffix we can create all $2^{i-2}$ different length-($i-1$) prefixes. That is, there are $2^{\pi-j-1}\cdot2^{i-2} = 2^{\pi+i-j-3}$ additional codesets for any $1<i<\pi$, $i-2<j<\pi$ and hence
                    \begin{align*}
                        \sum_{i=2}^{\pi}\sum_{j=i-1}^{\pi-1}2^{\pi+i-j-3} 
                        & = \sum_{i=2}^{\pi}(2^{\pi-1}-2^{i-2})\\
                        & = 2^{\pi-1}(\pi-1)+1-2^{\pi-1}\\
                        & = 2^{\pi-1}(\pi-2)+1
                    \end{align*}
                    additional codesets in total. 
                    \item[(b)]  If $i=j+2$  then all arguments from Case 4(a) still hold, however, since ${\sigma_{i-1}=\sigma_{j+1}}$ it can be observed that for $$\lambda = \sigma_i\cdots\sigma_{\pi}\sigma_1\cdots\sigma_j,$$ the codeset $$\Lambda = \{\lambda\}\cup \{\sigma_1,\ldots,\sigma_{i-2}, \sigma_{i},\sigma_{i+1},\ldots ,\sigma_\pi \}$$ yields a labeling in which all symbols of $S$ are labeled except for the symbol $\sigma_{i-1}$, in both of its occurrences in $S$. That is, $$\Lambda(S)_{[1;\pi]}=\Lambda(S)_{[\pi+1;\pi]} = 1^{i-2}01^{\pi-i+1}.$$ As this labeling was already considered in the symmetric case (using a different, $S$-equivalent codeset), for any $3\le i\le \pi$ we have  $2^{\pi+i-j-3}-1=2^{\pi-1}-1$ new codesets, and ${(\pi-2)(2^{\pi-1}-1)}$ additional codesets in total.
                    \item[(c)] If $i=j+3$ then, as in the previous case, (\ref{eq: symbol intersection}) still holds, however, in this case, we have that  $\sigma_{i-1}\in \{\sigma_{j+2},\ldots, \sigma_{\pi}\}$ and $\sigma_{j+1}\in\{\sigma_1,\ldots,\sigma_{i-2}\}$. Since we are avoiding labels that contain these two symbols, we cannot use them as part of the additional labels. Hence, we obtain $2^{\pi-j-1-1}\cdot 2^{i-2-1} = 2^{\pi+i-j-5}=2^{\pi-2}$  codesets. However, one of these codesets is not new as it satisfies the condition $\Lambda(S)_{[1;\pi]}=\Lambda(S)_{[\pi+1;\pi]}$. Namely, for $$\lambda = \sigma_i \cdots\sigma_\pi\sigma_1\cdots\sigma_j,$$ the codeset $$\Lambda=\{\lambda\}\cup\{\sigma_1, \ldots,\sigma_{i-3}\}\cup\{\sigma_{j+3},\cdots, \sigma_\pi\}$$ satisfies 
                    \begin{align*}
                    \Lambda(S) & = 1^{i-3}001^{|\lambda|}001^{\pi-j-2} \\ & = 1^{i-3}001^{\pi-i+1+j}001^{\pi-j-2}\\ &
                    = 1^{i-3} 0 0 1^{\pi-i+1} 1^{j}00 1^{\pi-j-2}\\
                    & = 1^j 0 0 1^{\pi-j-2} 1^j 0 0 1^{\pi-j-2},
                    \end{align*}
                    where the last equality follows since $i=j+3$. To see that this is the only codeset for which $\Lambda(S)_{[1;\pi]}=\Lambda(S)_{[\pi+1;\pi]}$ note that in any other such codeset, at least one of the symbols in $\{\sigma_1, \ldots,\sigma_{i-3}\}\cup\{\sigma_{j+3},\cdots, \sigma_\pi\}$ is not selected, and hence this symbol is labeled in the occurrence that corresponds to $\lambda$ but not labels in its other occurrence. Thus, the number of new codesets obtain for this case is, $2^{\pi-2}-1$ for any $4\le i\le\pi$ and $(\pi-3)(2^{\pi-2}-1)$ additional codesets in total.
                    \item[(d)] Otherwise, we have that $i> j+3$ which implies that (\ref{eq: symbol intersection}) does not hold. In particular, it holds that \begin{align*}
                        \abs{\{\sigma_1,\ldots,\sigma_{i-2}\}\cap\{\sigma_{j+2},\ldots, \sigma_\pi\}} =  i-j-3.
                    \end{align*}
                    This implies that we cannot choose any binary pattern we want for the prefix and suffix independently, as any symbol in this intersection is either labeled or not, in both halves of $\Lambda(S)$. That is, we can generate $2^{i-j-3}$ different patterns which correspond to the symbols in this intersection. As for the rest of the symbols, we note that in this case $\sigma_{j+1}\in\{\sigma_1,\ldots,\sigma_{i-2}\}$ and as we are avoiding labels that contain it, we left with $2^{i-2 -1- (i-j-3)}=2^{j}$ distinct options for the prefix. Similarly, we have that $\sigma_{i-1}\in \{\sigma_{j+2},\ldots, \sigma_{\pi}\}$ and there are $2^{\pi-j-1-1-(i-j-3)}=2^{\pi-i+1}$ options for the suffix. That is, the total number of codeset for this case is 
                    $2^{i-j-3}\cdot 2^{j}\cdot 2^{\pi-i+1} = 2^{\pi-2}$. However, similar to the previous case, we need to account for the codesets that are $S$ equivalent to codesets we already counted in the symmetric case. Here, to obtain symmetry we must again use all symbols that appear in $\lambda$ so their second occurrence will be labels. As for the other $i-j-3$ symbols in the intersection, each of them is either labeled in both halves of $\Lambda(S)$ or not, i.e., their selection does not affect the symmetry. Thus, the number of new codesets is $2^{\pi-2} - 2^{i-j-3}$ for any $5\le i\le \pi$, $1\le j\le i-4$, and 
                    \begin{align*}
                        &\sum_{i=5}^\pi  \sum_{j=1}^{i-4} \left(2^{\pi-2}  - 2^{i-j-3}\right) \\ & = \sum_{i=5}^\pi \left(2^{\pi-2}(i-4) - 2^{i-3}+2\right)\\
                        &= \frac{2^{\pi-2}(\pi^2\hspace{-.4ex}-\hspace{-.4ex}7\pi\hspace{-.4ex}+\hspace{-.4ex}12)}{2} \hspace{-.4ex} -\hspace{-.4ex} (2^{\pi-2}-4) \hspace{-.4ex}+\hspace{-.4ex} 2(\pi-4)\\
                        &= 2^{\pi-3}(\pi^2-7\pi+12) - 2^{\pi-2}+2\pi-4\\
                        & = 2^{\pi-3}(\pi^2-7\pi+10) +2\pi-4
                    \end{align*}
                \end{enumerate}
            \end{enumerate}
        \end{itemize}
        To summarize, we have $2^\pi$ codesets that correspond to symmetric labeling and
        \begin{align*}
            &\ \ \ \ \ 2(2^{\pi-1}-1) + 2^{\pi-1}(\pi-2)+1+(\pi-2)(2^{\pi-1}-1)\\
            &\ \ \ \ \ + (\pi-3)(2^{\pi-2}-1) + 2^{\pi-3}(\pi^2-7\pi+10) +2\pi-4\\
            &= 2^\pi-2+(\pi-2)(2^\pi-1)+1+2^{\pi-2}(\pi-3)\\
            & \ \ \ \ -(\pi-3)+ 2^{\pi-3}(\pi^2-7\pi+10) +2\pi-4\\
             &= 2^\pi +2^\pi(\pi-2)-(\pi-2)+2^{\pi-2}(\pi-3)\\
            & \ \ \ \ -(\pi-3)+ 2^{\pi-3}(\pi^2-7\pi+10) +2\pi-5\\
            & = 2^\pi + 2^\pi(\pi-2)+2^{\pi-2}(\pi-3)+2^{\pi-3}(\pi^2-7\pi+10)\\
            & = 2^{\pi-3}(8+8\pi-16+2\pi-6+\pi^2-7\pi+10)\\
            & = 2^{\pi-3}(\pi^2+3\pi-4)
        \end{align*}
        codesets that correspond to non-symmetric labeling, which sums up to
        $$
        2^\pi + 2^{\pi-3}(\pi^2+3\pi-4) = 2^{\pi-3}(\pi^2+3\pi+4).
        $$

Assume now that $n=r\pi$ for some integer $r\ge3$. We first note that all the complete intermediate periods have the same labeling. Hence, any labeling of $S$ is determined by the labelings of its first and last periods, together with the labeling of one intermediate period. Furthermore, the labelings of the first and last periods form
one of the pairs counted above for the case $n=2\pi$.

Consider first a nonsymmetric pair in Cases 4(b) -- 4(d), i.e., a pair corresponding to $i\ge j+2$. The label
\[
    \lambda
    =
    \sigma_i\cdots\sigma_\pi\sigma_1\cdots\sigma_j
\]
leaves the positions $j+1,\ldots,i-1$ unlabeled in every intermediate period. On the other hand, the longer label
\[
    \lambda'
    =
    \sigma_i\cdots\sigma_\pi
    \sigma_1\cdots\sigma_\pi
    \sigma_1\cdots\sigma_j
\]
has the same effect on the first and last periods, while its occurrences overlap and label all the intermediate periods. Thus, each pair counted in Cases 4(b) -- 4(d) gives one additional labeling when $n\ge3\pi$.

The number of such nonsymmetric pairs was calculated above and is equal to
\begin{align*}
N_{\mathsf{gap}}
={}&
(\pi-2)(2^{\pi-1}-1)
+(\pi-3)(2^{\pi-2}-1)\\
&\quad
+2^{\pi-3}(\pi^2-7\pi+10)+2\pi-4\\
={}&
2^{\pi-3}(\pi^2-\pi-4)+1.
\end{align*}

We also need to consider the symmetric pairs that were included among the first $2^\pi$ labelings. Let $X\in\{0,1\}^\pi$ be the common labeling of the first and last periods. If the first and last entries of $X$ are both one, then $X$ can be obtained using a label that crosses a period boundary. Consequently, we can obtain both
\[
    X^r
    \qquad\text{and}\qquad
    X(1^\pi)^{r-2}X.
\]
These two labelings are distinct unless $X=1^\pi$. Therefore, the symmetric pairs contribute
\[
    2^{\pi-2}-1
\]
additional labelings.

No other labeling of the intermediate periods is possible. Indeed, if all the labels crossing a period boundary have length at most $\pi$, then the same gap pattern is repeated in every intermediate period. If at least one such label has length greater than $\pi$, then its consecutive occurrences overlap, and all the intermediate periods are labeled entirely.

It follows that, for $n\ge3\pi$, the number of additional labelings
compared with the case $n=2\pi$ is
\begin{align*}
N_{\mathsf{extra}}
&=
N_{\mathsf{gap}}+2^{\pi-2}-1\\
&=
2^{\pi-3}(\pi^2-\pi-2).
\end{align*}
Therefore,
\begin{align*}
M(S)
&=
2^{\pi-3}(\pi^2+3\pi+4)
+
2^{\pi-3}(\pi^2-\pi-2)\\
&=
2^{\pi-2}(\pi^2+\pi+1),
\end{align*}
which completes the proof.
    \hfill\quad $\blacksquare$

\noindent \textit{Proof of~\autoref{lem: eta}:}
For convenience, let $A\triangleq \cR_\ell(n)$; 
     we start by partitioning $A$ into two disjoint sets,
    \begin{align*}
        A^0& =\left\{(x_1,\ldots,x_n)\in A:\ x_1=0\right\}\\
        A^1& =\left\{(x_1,\ldots,x_n)\in A:\ x_1=1\right\}.
    \end{align*}
    We have that
    $|A^0| =  \sum_{r=1}^n  |A^0_{r}|$,
    where $A^0_{r}$ is the set of all sequences in $A^0$ that have exactly $r$ runs (of \emph{zeros} and \emph{ones}). Since the first run of any sequence in $A^0$ is a run of \emph{zeros}, any sequence in $A^0_{r}$ contains $r_1 = \lceil(r-1)/2\rceil$ runs of \emph{ones} and $r_0=r-r_1 =\lfloor(r+1)/2\rfloor $ runs of \emph{zeros}.
    Additionally, each run of \emph{zeros} must contain at least a single \emph{zero} and each run of \emph{ones} must contain at least $\ell$ \emph{ones}. Note that $|A_r^0|$ is equal to  the number of ways to select the lengths of these $r$ runs. As we have total of $n$ symbols, and we need to have at least a single \emph{zero} in each run of \emph{zeros} and at least $\ell$ \emph{ones} at each run of \emph{ones}, we have $n-\ell r_1-r_0$ more bits that can be assigned to the different runs without restrictions.  Let $\CC{m}{b}\triangleq \binom{m+b-1}{b}=\binom{m+b-1}{m-1}$ denote the number of $b$-element combinations of $m$ objects, with repetitions. We have that  
    \begin{align*}
        |A_{r}^0| & = \hspace{-2pt}\CC{r}{n-\ell r_1-r_0} \hspace{-2pt}= \hspace{-2pt}\CC{r}{n-\lceil(r-1)/2\rceil(\ell-1) -r } \\
        & = \hspace{-2pt}\binom{n-\lceil(r-1)/2\rceil(\ell-1)- 1}{r-1}
    \end{align*}
    and since the number of runs can be at least $1$ and at most $n$,
    \begin{align*}
        |A^0| =  \sum_{r=1}^n  \binom{n-\lceil(r-1)/2\rceil(\ell-1)- 1}{r-1},
    \end{align*}
    where $\binom{m}{b}=0$ if $m<b$ or $m<0$.

Similarly, we have that 
\begin{align*}
    |A^1| & = \sum_{r=1}^n \CC{r}{n-\lfloor(r+1)/2\rfloor(\ell-1)-r}\\
    & = \sum_{r=1}^n \binom{n-\lfloor(r+1)/2\rfloor(\ell-1)-1}{r-1}
\end{align*}
and thus 
\vspace{-1.5ex}
\begin{align*}
    |\cR_\ell(n)|= & \sum_{r=1}^n
    \binom{n-\lceil\frac{r-1}{2}\rceil(\ell-1)- 1}{r-1} \\ & +\sum_{r=1}^n \binom{n-\lfloor\frac{r+1}{2}\rfloor(\ell-1)-1}{r-1}.
\end{align*}
Assume $n$ is odd then 
\begin{align*}
    \sum_{r=1}^n&\binom{n-\lfloor\frac{r+1}{2}\rfloor(\ell-1)-1}{r-1} =  \sum_{t=1}^{\frac{n-1}{2}}\binom{n-t(\ell-1)-1}{2t-2}  \\
    &\ \ \ \ +\sum_{t=1}^{\frac{n-1}{2}} \binom{n-t(\ell-1)-1}{2t-1} + \underbrace{\binom{n-\frac{n+1}{2}(\ell-1)-1}{\frac{n+1}{2}-1}}_{=0}\\
    & = \sum_{t=1}^{\frac{n-1}{2}}\left(\binom{n-t(\ell-1)-1}{2t-2} + \binom{n-t(\ell-1)-1}{2t-1}\right)\\
    & = \sum_{t=1}^{\frac{n-1}{2}}\binom{n-t(\ell-1)}{2t-1},
\end{align*}
where the last equality follows from the identity $\binom{n-1}{k-1}+\binom{n-1}{k}=\binom{n}{k}$.
Similarly, 
\begin{align*}
    \sum_{r=1}^n &\binom{n-\lceil\frac{r-1}{2}\rceil(\ell-1)-1}{r-1} =  \\
    &\ \ \ \ 1 + \sum_{t=1}^{\frac{n-1}{2}} \left(\binom{n-t(\ell-1)-1}{2t-1} + \binom{n-t(\ell-1)-1}{2t}\right)\\
    & = 1 + \sum_{t=1}^{\frac{n-1}{2}}\binom{n-t(\ell-1)}{2t}.
\end{align*}
Hence, for odd $n$ we have that 
\begin{align*}
    |\cR_\ell(n)| & = 1 + \sum_{t=1}^{\frac{n-1}{2}}\left(\binom{n-t(\ell-1)}{2t-1}+ \binom{n-t(\ell-1)}{2t}\right)\\
    & = 1 + \sum_{t=1}^{\frac{n-1}{2}}\binom{n-t(\ell-1)+1}{2t}\\
    & = \sum_{t=0}^{\frac{n-1}{2}}\binom{n-t(\ell-1)+1}{2t}.
\end{align*}
Using the same technique, it can be shown that for even $n$
$$
|\cR_\ell(n)| = \sum_{t=0}^{\frac{n}{2}}\binom{n-t(\ell-1)+1}{2t}.
$$
Finally, by observing that $n-t(\ell-1)+1\ge 2t$ if and only if $\frac{n-1}{\ell+1}\ge t$, and hence $\binom{n-t(\ell-1)+1}{2t}=0$ for any $t>\frac{n-1}{\ell+1}$, we have that for any two integers $n$ and $\ell$ we have that 
\[
|\cR_\ell(n)| =  \sum_{t=0}^{\lfloor \frac{n+1}{\ell+1}\rfloor}\binom{n-t(\ell-1)+1}{2t}.
\vspace{-1.5ex}\]
    \hfill\quad $\blacksquare$
\end{document}